\newif\iflong 
\newif\ifshowcomments 
\newif\ifappendix 
\newif\ifhal

\ifdefined\shortversion
\else
\longtrue
\haltrue
\fi

\ifhal
\documentclass[a4paper,UKenglish]{article}
\else
\documentclass[a4paper,UKenglish,cleveref,autoref, thm-restate,numberwithinsect,anonymous]{lipics-v2021}
\fi

\usepackage[utf8]{inputenc}
\usepackage[T1]{fontenc}
\usepackage{graphicx} 

\usepackage{amsmath}
\usepackage{amssymb}
\usepackage{stmaryrd}
\usepackage{tikz}
\usepackage{tkz-euclide}
\usepackage[]{hyperref}
\usepackage{nicefrac}
\usepackage{dsfont}
\usepackage{todonotes}
\usepackage{xspace}
\usepackage{cite}

\ifhal
\usepackage{authblk}
\usepackage{a4wide}
\usepackage{cleveref}
\newcommand{\authorrunning}[1]{}
\newcommand{\Copyright}[1]{}
\newcommand{\keywords}[1]{\def\thekeywords{#1}}
\newtheorem{theorem}{Theorem}

\newtheorem{lemma}{Lemma}
\newtheorem{corollary}{Corollary}
\newtheorem{proposition}{Proposition}
\newtheorem{claim}{Claim}
\newtheorem{definition}{Definition}
\def\Box{\hbox{\hskip 1pt \vrule width 4pt height 8pt depth 1.5pt \hskip 1pt}}
\newenvironment{proof}{\medskip\noindent\textbf{Proof.}}{{}\hfill$\Box$\\}
\newenvironment{claimproof}{\medskip\noindent\textit{Proof.}}{{}\hfill$\lhd$\\}
\fi

\ifappendix
\usepackage[absolute]{textpos}
\fi

\usetikzlibrary{calc}

\newcommand{\EARLY}{\textsc{Foremost-path TGR}\xspace}
\newcommand{\REARLY}{\textsc{Ranged-Foremost-path TGR}\xspace}
\newcommand{\RNSEARLY}{\textsc{Ranged-NS-Foremost-path TGR}\xspace}
\newcommand{\FAST}{\textsc{Fastest-path TGR}\xspace}
\newcommand{\NSFAST}{\textsc{NS-Fastest-path TGR}\xspace}
\newcommand{\SHORT}{\textsc{Shortest-path TGR}\xspace}
\newcommand{\NSSHORT}{\textsc{NS-Shortest-path TGR}\xspace}
\newcommand{\PFAST}{\textsc{Periodic Fastest-path TGR}\xspace}
\newcommand{\PNSFAST}{\textsc{Periodic NS-Fastest-path TGR}\xspace}
\newcommand{\MCC}{\textsc{Multicolored Clique}\xspace}
\newcommand{\Reach}{\textsc{Reachability Graph Realization}\xspace}
\newcommand{\NReach}{\textsc{Non-Strict Reachability Graph Realization}\xspace}
\newcommand{\SAT}{\textsc{SAT}\xspace}

\newcommand{\FoTGR}{\textsc{Foremost-path TGR}\xspace}
\newcommand{\NSFoTGR}{\textsc{NS-Foremost-path TGR}\xspace}
\DeclareMathOperator{\Fo}{Foremost}
\DeclareMathOperator{\Fa}{Fastest}
\DeclareMathOperator{\Sh}{Shortest}
\DeclareMathOperator{\NSFo}{NS-Foremost}
\DeclareMathOperator{\NSFa}{NS-Fastest}
\DeclareMathOperator{\NSSh}{NS-Shortest}
\DeclareMathOperator{\Prescribed}{Prescribed}
\DeclareMathOperator{\Periodic}{Periodic}

\newcommand{\NSECompat}{NSEdgeCompat}

\newcommand{\Oh}{\mathcal{O}}
\newcommand{\tOh}{\widetilde{\mathcal{O}}}
\newcommand{\mg}{\mathcal{G}}
\newcommand{\G}{\mathcal{G}}
\newcommand{\T}{\mathcal{T}}
\newcommand{\M}{M}
\newcommand{\TG}{\mathcal{TG}}

\ifshowcomments
\newcommand{\todom}[2][]{\todo[#1,color=green!50]{#2}}
\newcommand{\todomi}[1]{\todom[inline]{#1}}
\newcommand{\lvtodo}[1]{\todo[color=blue!50]{#1}}
\newcommand{\lv}[1]{\textcolor{blue}{LV: #1}}
\newcommand{\jctodo}[1]{\todo[color=red!50]{#1}}
\newcommand{\jc}[1]{\textcolor{red}{#1}}

\newcommand{\lvnew}[1]{\textcolor{blue}{#1}}
\else
\newcommand{\todom}[2][]{}
\newcommand{\todomi}[1]{}
\newcommand{\lvtodo}[1]{}
\newcommand{\lv}[1]{}
\newcommand{\jctodo}[1]{}
\newcommand{\jc}[1]{}

\newcommand{\lvnew}[1]{#1}
\fi

\newcommand{\prob}[3]{
\smallskip

\noindent \textsc{#1}:

\noindent \textbf{Input:} #2

\noindent \textbf{Question:} #3

\smallskip
}

\authorrunning{J.~Cauvi, N.~Morawietz, and L.~Viennot}
\Copyright{Justine Cauvi, Nils Morawietz, and Laurent Viennot}

\usepackage{algorithm}
\usepackage[noend]{algpseudocode}

\ifhal
\else
\usepackage{amsthm}
\fi

\newcommand{\ol}[1][x]{\overline{#1}}

\ifhal

\title{Foremost, Fastest, Shortest: Temporal Graph Realization under Various Path Metrics%
\thanks{Supported by the French ANR, projects ANR-22-CE48-0001 (TEMPOGRAL)  and ANR-24-CE48-4377 (GODASse).}}
\author[1]{Justine Cauvi}
\affil[1]{\'Ecole Normale Supérieure de Lyon, Lyon, France\\Inria, DI ENS, Paris, France}
\author[2]{Nils Morawietz}
\affil[2]{LaBRI, Université de Bordeaux, France\\Institute of Computer Science, Friedrich Schiller University Jena,  Germany}
\author[3]{Laurent Viennot}
\affil[3]{Inria, DI ENS, Paris, France}
\date{October 2025}
\keywords{network design, temporal paths, foremost paths, fastest paths, shortest paths, non-strict paths, periodic temporal graphs}

\else

\title{Foremost, Fastest, Shortest: Temporal Graph Realization under Various Path Metrics}
\author{Justine Cauvi}{\'Ecole Normale Supérieure de Lyon, Lyon, France\\Inria, DI ENS, Paris, France}{justine.cauvi@ens-lyon.fr}{https://orcid.org/0009-0001-5444-8877}{}
\author{Nils Morawietz}{LaBRI, Université de Bordeaux, France\\Institute of Computer Science, Friedrich Schiller University Jena,  Germany}{nils.morawietz@uni-jena.de}{https://orcid.org/0000-0002-7283-4982}{}

\funding{Supported by the French ANR, projects ANR-22-CE48-0001 (TEMPOGRAL)  and ANR-24-CE48-4377 (GODASse).}

\author{Laurent Viennot}{Inria, DI ENS, Paris, France}{}{}{}

\keywords{network design, temporal paths, foremost paths, fastest paths, shortest paths, non-strict paths, periodic temporal graphs}

\ccsdesc[500]{Theory of Computation~Design and analysis of algorithms}
\fi

\begin{document}

\maketitle 

\ifappendix
\begin{textblock}{15}(3,3) 
	\Huge\bf Appendix (full version)
\end{textblock}
\fi

\begin{abstract}
In this work, we follow the current trend on temporal graph realization, where one is given a property~$P$ and the goal is to determine whether there is a temporal graph, that is, a graph where the edge set changes over time, with property~$P$. 
We consider the problems where as property~$P$, we are given a prescribed matrix for the duration, length, or earliest arrival time of pairwise temporal paths.
That is, we are given a matrix~$D$ and ask whether there is a temporal graph such that for any ordered pair of vertices~$(s,t)$, $D_{s,t}$ equals the duration (length, or earliest arrival time, respectively) of any temporal path from~$s$ to~$t$ minimizing that specific temporal path metric.
For shortest and earliest arrival temporal paths, we are the first to consider these problems as far as we know.
We analyze these problems for many settings like: strict and non-strict paths, periodic and non-periodic temporal graphs, and limited number of labels per edge (that is, limited occurrence number per edge over time). 
In contrast to all other path metrics, we show that for the earliest arrival times, we can achieve polynomial-time algorithms in periodic and non-periodic temporal graphs and for strict and and non-strict paths.
However, the problem becomes NP-hard when the matrix does not contain a single integer but a set or range of possible allowed values.
As we show, the problem can still be solved efficiently in this scenario, when the number of entries with more than one value is small, that is, we develop an FPT-algorithm for the number of such entries.
For the setting of fastest paths, we achieve new hardness results that answers an open question by Klobas, Mertzios, Molter, and Spirakis [Theor. Comput. Sci. '25] about the parameterized complexity of the problem with respect to the vertex cover number and significantly improves over a previous hardness result for the feedback vertex set number.
When considering shortest paths, we show that the periodic versions are polynomial-time solvable whereas the non-periodic versions become NP-hard.
\end{abstract}

\ifhal
\textbf{Keywords: }\thekeywords
\clearpage
\fi

\ifdefined\shortversion
\vfill
\noindent {Due to space constraints, proofs of results marked with $\star$ are (partially) deferred to a clearly marked appendix.}
\clearpage
\fi

\section{Introduction}
Graph realization problems have been studied since the 1960s and consists in finding a static graph that satisfies a desired property~$P$ or answering that no such graph exists. 
The earliest example of such a problem is the case of degree sequence realization, where one is given a non-decreasing sequence~$(d_1,\dots, d_n)$ of natural numbers, and one asks whether there is an undirected graph~$G$ with vertex set~$\{1,\dots,n\}$, such that vertex~$i$ has degree exactly~$d_i$.
This problem was introduced by Erd\H{o}s and Gallai~\cite{erdos1960} and has been generalized until recently (see, e.g., \cite{BarNoyCPR2020} when ranges are given for each degree).
In another early graph realization problem by Hakimi~and~Yau~\cite{hakimi1965}, one is given an~$n\times n$ distance matrix~$D$ and asks whether there is a static graph on~$n$ vertices, where for each ordered pair~$(i,j)$ of vertices, the shortest path from~$i$ to~$j$ has length exactly~$D_{i,j}$.
Since then, graph realization problems were considered in many variations and for many other desirable properties to realize.
Recently, Klobas, Mertzios, Molter and Spirakis~\cite{KMMS24} lifted graph realization problems to the realm of temporal graphs motivated by the realization problem for distance matrices on static graphs.
Here, a~\emph{temporal graph}~$\mg$ is a finite sequence~$(G_1,\dots,G_L)$ of static graphs that are all defined over the same vertex set.
Temporal graphs are a valuable tool to model and analyze the behavior of real world dynamic networks~\cite{casteigts2012time}.
In the problem introduced by Klobas et al.~\cite{KMMS24}, one asks whether our desired property is fulfilled by some temporal graph.\footnote{It is worth mentioning that Göbel, Cerdeira and Veldman~\cite{GCV91} considered a connectivity problem that can also be seen as a temporal graph realization problem.
} 
They introduced the following problem: 
\prob{\FAST}
{An~$n\times n$ distance matrix $D$.}
{Is there a temporal graph $\mg$ with $n$ vertices such that for any ordered pair~$(s,t)$ of vertices, the fastest temporal path\footnote{A (strict) temporal path can start at any time step and is allowed to traverse only a single edge per snapshot, and its duration is the difference between the starting time and the arrival time (see~\Cref{sec prelim}).} from~$s$ to~$t$ has duration~$D_{s,t}$?}
More precisely, the authors considered this problem\footnote{They analyzed the problem under the name \textsc{Simple (Periodic) Temporal Graph Realization}.}  
where only a single label per edge is allowed and where the temporal graph is periodic, that is, where the same edges repeat every~$\Delta$ time steps for some period $\Delta>0$.
 They showed that the problem is NP-hard and they exhibited an FPT algorithm parameterized by the feedback edge number of the underlying graph\footnote{The underlying graph is uniquely defined via the vertex pairs~$(s,t)$ for which~$D_{s,t}=D_{t,s} = 1$.} while they showed that it is W[1]-hard when parameterized by the feedback vertex set number.
  Erlebach, Morawietz and Wolf~\cite{EMW24} generalized the problem by allowing each edge to appear up to~$\ell$ times per period and proved that this remains NP-hard even on underlying graphs that are trees for~$\ell = 5$. 
   In~\cite{mertzios2025}, the authors studied the problem where upper bounds on the fastest paths are given and the underlying graph is a tree.
   This was further considered for directed graphs by Meusel, Müller-Hannemann and Reinhardt~\cite{MMR25}.
   Further recent papers on temporal graph realization include: designing a temporal graph for which the fastest path should not have duration more than~$\alpha$ times the real distance~\cite{mertzios2025-2}, designing a temporal graph which should have a prescribed reachability relation between the vertices~\cite{EMM25}, and designing a temporal graph for which all pairs of agents can pairwise reach each other via strict temporal paths, with one label per edge, while the degree sequence of the underlying graph is prescribed~\cite{CDM25}.
   All these problems are motivated both from a design perspective, where we aim to design a network with a desired behavior, or from a verification perspective, where we want to verify that the behavior of our network is correct or at least plausible.   
From the perspective of temporal network design problems, the field is even more vibrant (see, e.g.,~\cite{kempe2000connectivity,mertzios2013temporal,MertziosMS19,enright2021deleting,zschoche2020complexity,deligkas2022optimizing,akrida2017complexity,klobas2024complexity}).


\subparagraph{Our Results.}
In this work, we extend the previous work on \FAST to the non-strict case (that is, where arbitrary many edges per snapshot can be traversed). Furthermore, we consider the other most famous temporal path metrics $\Fo$ and $\Sh$ (formally defined in~\Cref{sec prelim}).
Roughly speaking, \EARLY requires that the earliest arrival time at~$t$, minimized over all temporal $st$-path, is equal to~$D_{s,t}$ for each vertex pair~$(s,t)$, whereas \SHORT requires that the number of edges, minimized over all temporal $st$-path, is equal to~$D_{s,t}$ for each vertex pair~$(s,t)$. 
Our main results are as follows:
\begin{enumerate}
\item In~\Cref{sec early} we show that all considered version (strict/non-strict, periodic/non-periodic) of \EARLY are polynomial-time solvable, if we are allowed to put an arbitrary number of labels per edge.  
This is surprising, as all other previously mentioned temporal graph realization problems turn out to be NP-hard (with one exception~\cite{CDM25}).
In particular, we show that all our algorithms produce a labeling with
\iflong
at most~$n$ labels per edge\lvtodo{not stated in the sequel} and 
\fi
at most~$n^2$ time labels in total, if dealing with a realizable instance. This is somehow tight as some realizable matrices do require $\Omega(n^2)$ time labels%
\iflong 
~or have edges that need at least~$\Omega(n)$ time labels.
\else .
\fi
\item To show the limitations of the tractability, we show that \EARLY becomes NP-hard if (i)~we are only allowed to assign one label per edge or (ii)~the matrix~$D$ contains for each vertex pair more than one entry and we are to choose which of these possible values we want to realize.
For the latter problem version, we present a single exponential FPT-algorithm when parameterized by the number of entries in~$D$ that have more than one possible value.
\item In~\Cref{sec fast} we consider \FAST.
Among other results, we answer open questions by Klobas et al.~\cite{KMMS24} and Erlebach et al.~\cite{EMW24} about the parameterized complexity of the problem for the vertex cover number.
We show that the problem is W[1]-hard when parameterized by this parameter plus the largest entry in the matrix.
This result improves significantly over the previous known parameterized hardness result for the feedback vertex set number; in terms of the parameter, the construction, and the length.
\item 
In~\Cref{sec short} we consider \SHORT and show that the problem for both the strict and the non-strict case are NP-hard, but becomes trivial when considering a periodic temporal graph.
\end{enumerate}
Finally, in~\Cref{sec conc}, we conclude with some open questions for future work.

\iflong
\todom{Insert a "Technical highlights" subparagraph?}
\fi

\section{Preliminaries}\label{sec prelim}

For natural numbers~$1\leq i \leq j$, we let~$[i,j]:= \{i, i+1, \dots, j\}$ and we define~$[j]:= [1,j]$.

\textbf{(Static) graphs.}
An (undirected) \emph{graph} $G=(V,E)$ is defined by its vertex set $V$ and its edge set $E\subseteq \binom{V}{2}$. Any pair $\{u,v\}\in E$ is called an \emph{edge} between vertices $u$ and $v$. We also say that $u$ and $v$ are \emph{neighbors} when $\{u,v\}\in E$. 
For a vertex set~$S\subseteq V$, we define the \emph{subgraph of~$G$ induced by~$S$} as~$G[S] := (S,E \cap \binom{S}{2})$. 
A graph $H=(V',E')$ is a \emph{subgraph} of $G$ if~$V'\subseteq V$ and~$E'\subseteq E$.
A \emph{$uv$-walk} is a sequence $P=(v_0=u,v_1,\ldots,v_k=v)$ of vertices such that $\{v_{i-1},v_i\}$ is an edge for every $i\in [k]$. We then say that $P$ is a walk from $u$ to $v$ in $G$. Such a walk is called a path if the vertices $v_0,\ldots,v_k$ are pairwise distinct.

\textbf{Temporal graphs.}
A \emph{temporal graph} is defined by a pair $\mg=(G,\lambda)$ where $G=(V,E)$ is a graph, and $\lambda:E\rightarrow 2^{\mathbb{N}_{>0}}$ is a labeling that associates to each edge $e\in E$ the set $\lambda(e)$ of (positive) times when $e$ appears. The graph $G$ is called the \emph{underlying graph} of $\mg$ and the labeling $\lambda$ is called the \emph{time labeling} of $\mg$. When $\mg$ is clear from the context, we let $n=|V|$ denote its number of vertices. A \emph{temporal edge} is a couple $(\{u,v\},\tau)$ such that $\{u,v\}\in E$ and $\tau\in \lambda(\{u,v\})$. Its \emph{appearance time} is $\tau$. Given a time label $\tau\in\mathbb{N}_{>0}$, we define the set $E_\tau=\{e\in E : \tau \in \lambda(e)\}$ of edges appearing at time $\tau$, and $G_\tau=(V,E_\tau)$ is called the \emph{snapshot} of $\mg$ at time $\tau$. The size of a temporal graph can be measured by its number $\sum_{e\in E}|\lambda(e)|$ of time labels.
A temporal graph is said to be \emph{$\Delta$-periodic} if each edge appears periodically with period $\Delta\in\mathbb{N}_{>0}$, that is $\tau\in\lambda(e)$ if and only if $\tau\mod\Delta \in\lambda(e)$. Such a temporal graph is represented by its list of temporal edges up to time $\Delta$.
%
In the following, we always assume that the vertices of a temporal graph $\G$ are numbered from $1$ to~$n$. We assume without loss of generality that its vertex set is $V=[n]$. We let $\TG$ denote the set of all such temporal graphs. 

\textbf{Temporal paths.}
A \emph{strict} (resp. \emph{non-strict}) \emph{temporal $uv$-walk} is a walk $P=(v_0=u,v_1,\ldots,v_k=v)$ in $G$ which is associated to time labels $(\tau_1,\ldots,\tau_k)$ such that $\tau_i\in\lambda(\{v_{i-1},v_i\})$ for each $i\in [k]$ and $\tau_1<\cdots <\tau_k$ (resp. $\tau_1\le\cdots \le\tau_k$). 
Equivalently, it can be defined as the sequence of temporal edges $(\{v_0,v_1\},\tau_1),\ldots,(\{v_{k-1},v_k\},\tau_k)$ satisfying $\tau_1<\cdots <\tau_k$ (resp. $\tau_1\le\cdots \le\tau_k$). 
It is said to have \emph{length} $k$, \emph{departure time} $\tau_1$, \emph{arrival time} $\tau_k$, and \emph{duration} $\tau_k-\tau_1+1$ (number of time steps spanned). 
It is called a \emph{temporal $uv$-path} if the vertices $v_0,\ldots,v_k$ are pairwise distinct. Note that a strict (resp. non-strict) temporal walk can always be transformed into a strict (resp. non-strict) temporal path by removing loops, and this can only reduce length, arrival time and duration.
By default, we consider strict temporal paths, and simply call them temporal paths. We specify non-strict for non-strict temporal paths.

\textbf{Temporal path metrics.} 
Classically, a temporal $uv$-path is said to be \emph{shortest}, \emph{foremost}, or \emph{fastest} if it has minimum length, arrival time, or duration, respectively among all temporal $uv$-paths. These notions indeed define some kind of metrics that we now formalize. A \emph{distance matrix} $D$ is any matrix of size $n\times n$ with values in $\mathbb{N}\cup\{\infty\}$ and that satisfies the following very loose notion of metric: $D_{uv}=0$ if and only if $u=v$, for all $u,v\in [n]$. It is not assumed that $D$ satisfies any other specific properties. 
\iflong
In particular, it may violate the triangle inequality.
\fi
A \emph{temporal path metric} is defined as a function $\M$ that associates a distance matrix $\M(\G)$ to any temporal graph $\G\in \TG$. Consider for example,
a \emph{cost function} $C$ that associates a positive cost in $\mathbb{N}_{>0}$ to any temporal path given as a sequence of temporal edges (independently of any temporal graph). It defines a \emph{temporal path metric} by associating to any temporal graph $\G\in \TG$ the matrix $C(\G)$ (resp. $\text{NS-}C(\G)$) such that $C(\G)_{uv}$ (resp. $\text{NS-}C(\G)_{uv}$)  is the minimum cost of a strict (resp. non-strict) temporal $uv$-path in $\G$. We define $C(\G)_{uv}:=\infty$ (resp. $\text{NS-}C(\G)_{uv}:=\infty$) if no strict (resp. non-strict) temporal $uv$-path exists and $C(\G)_{uv}:=0$ (resp. $\text{NS-}C(\G)_{uv}:=0$) if $u=v$. The foremost, fastest, and shortest notions are indeed associated to the following cost functions: arrival time, duration, and length, respectively. We let $\Fo$, $\Fa$, and $\Sh$ denote the corresponding temporal path metrics, respectively. We also let $\NSFo$, $\NSFa$, and $\NSSh$ denote the respective variants for non-strict temporal paths. For example, given a temporal graph $\G$, and two vertices $u,v$ in $\G$, $\Fo(\G)_{uv}$ is the earliest arrival time of a strict temporal $uv$-path in $\G$.

\textbf{Temporal graph realization.} 
Given an integer $n$, the temporal graph realization problem consists in finding a temporal graph with $n$ vertices that satisfies a given property. 
We assume that this property can be expressed as an input sequence $I\in\{0,1\}^*$ of bits, given that the vertices of the temporal graph are $1,\ldots,n$.
More precisely, we define a \emph{predicate} $P$ as a binary relation between the set $\TG$ of all temporal graphs and the set $\{0,1\}^*$ of all bit sequences, that is $P$ is a subset of $\TG\times \{0,1\}^*$. We then say that a temporal graph $\G$ \emph{satisfies} a sequence $I$ of bits for $P$ if $(\G,I)\in P$. We equivalently say that $P(\G,I)$ is \emph{satisfied}, or that $\G$ is a \emph{realization} of $I$ for $P$.
As a simple example, the \emph{lifetime} of a temporal graph can be tested by the predicate $P(\G,\Lambda):=\max(\cup_{e\in E} \lambda(e)) = \Lambda$ where $\G=(G,\lambda)$ and $\Lambda$ encodes an integer (that we denote also by $\Lambda$ with a slight abuse of notation). That is $P(\G,\Lambda)$ is satisfied when $\Lambda$ is the last appearance time of a temporal edge of $\G$.
Given a predicate $P$ we thus define the following (very general) problem.
\prob{$P$ Temporal Graph Realization ($P$ TGR)}
{A number $n$ and a sequence $I$ of bits.}
{Is there a temporal graph $\mg$ with $n$ vertices such that $P(\G,I)$ is satisfied?}

This paper focuses on temporal path metric realization. More precisely,
considering a temporal path metric $\M\in\{\Fo, \Fa, \Sh, \NSFo, \NSFa$, $\NSSh\}$, 
we define the \emph{$M$-temporal-path-metric} predicate, or \emph{$M$-path} for short, as $P(\G,I) := M(\G) = D$ where the sequence $I$ of bits encodes an $n\times n$ distance matrix $D$ where $n$ is the number of vertices of $\G$. For example, $\Fo$-path TGR is the following problem.
\prob{\FoTGR}
{A number $n$ and an $n\times n$ distance matrix $D$.}
{Is there a temporal graph $\mg$ with $n$ vertices such that $\Fo(\G)=D$?}

When considering an instance of $P$ TGR, we always let $n$ denote the associated number. Note that, in the case of an \textsc{$M$-path TGR} instance, the input has size $\Theta(n^2)$. For brevity, given an input sequence $I$ encoding a distance matrix $D$, a realization $\G$ of $I$ for $M$-path is simply called an \emph{$M$-realization} of $D$.
If further the metric~$M$ is clear from the context, we simply say that 
$\mg$ is a~\emph{realization} of~$D$ or that $\G$ \emph{realizes} $D$. 

We also combine $M$-path predicates with additional requirements. In particular,
when the input $I$ encodes a prescribed (static) graph $G_p$ with $n$ nodes, or a period $\Delta$, 
we define the following additional predicates, respectively:
\begin{itemize}
	\item $\Prescribed(\G,G_p) := $ the underlying graph of $\G$ is a subgraph of $G_p$,
	\item $\Periodic(\G,\Delta) := $ $\G$ is $\Delta$-periodic.
\end{itemize} 
Given two predicates $P,Q$, the problem $P$ $Q$ TGR asks whether there exists a temporal graph satisfying both $P$ and $Q$ for a given pair of inputs for $P$ and $Q$.
For example, \textsc{Periodic \FoTGR} asks, given an $n\times n$ distance matrix $D$ and a period $\Delta$, whether there exists a $\Delta$-periodic temporal graph that is a $\Fo$-realization of $D$.
We also sometimes refer to $P$ $Q$ TGR as $Q$ TGR where $P$ is required, especially if $P$ is specified in plain text without defining a formal name for it.

\section{Foremost paths}\label{sec early}

We first consider the $\Fo$-path TGR problem. Recall that, given an $n\times n$ distance matrix $D$, it consists in checking whether there exists a temporal graph $\G$ whose foremost matrix is $D$, i.e. $\Fo(\G)=D$. We also consider its non-strict variant (using $\NSFo$), and combinations with Prescribed and Periodic additional requirements.
Similar results can be obtained for latest departure using a time-reversal argument.
\iflong
More precisely, the latest-departure temporal path metric is associated to the cost function that assigns the opposite of the departure time to a temporal path. The time-reversal $\G'$ of a temporal graph $\G$ is obtained by changing each time label $\tau$ by $\Lambda+1-\tau$ where $\Lambda$ is the maximum time label. As a result any latest-departure temporal $uv$-path in $\G$ corresponds to a foremost temporal $vu$-path in $\G'$ where edges are traversed in reverse order.
\fi
We often implicitly assume that we are given an $n\times n$ distance matrix $D$.
\iflong
\else
We first show that these main variants of the problem can be solved in polynomial time.
\fi

\subsection{Polynomial time algorithms}\label{se:poltime-foremost}

\iflong
We first show that these main variants of the problem can be solved in polynomial time.
\fi

\paragraph*{Strict Foremost paths}
In the strict setting, we can indeed state the following.

\begin{theorem}\label{th:foremost}
	\FoTGR can be solved in $\Oh(n^3\log n)$ time and $\Oh(n^2)$ space. Furthermore, if dealing with a realizable instance, a realization with at most $n^2$ time labels can be computed with same complexity.
\end{theorem}

Note that the above complexity can be expressed as $\tOh(N^{3/2})$ where $N=\Theta(n^2)$ is the size of the input.
We also state that the above result is somehow tight in terms of the number of time labels as follows.


\iflong
\begin{proposition}
\else
\begin{proposition} 
\fi\label{prop:lower-bound}
There exists a family $(D^n)_{n\in \mathbb{N}_{>0}}$ of $n\times n$ distance matrices that are $\Fo$-realizable and such that any realization requires $\Omega(n^2)$ time labels.
\end{proposition}

\begin{proof}
	Consider the temporal graph $\G^n$ with $n$ vertices, whose underlying graph is a star rooted at $1$, and where, for all $v>1$, edge $\{1,v\}$ appears at times $nv,n(v+1)+v,n(v+2)+v,\ldots,n^2+v$. Its foremost matrix $D^n=\Fo(\G^n)$ then satisfies $D^n_{uv}=nu+v$ for $u>v>1$, $D^n_{uv}=nv$ for $1<u<v$, and $D^n_{1v}=D^n_{v1}=nv$ for $1<v$. These temporal graphs thus define a family $(D^n)_{n\in \mathbb{N}_{>0}}$ of distance matrices that are foremost realizable and that have $\Omega(n^2)$ pairwise distinct entries. This implies that any foremost realization of such a matrix $D^n$ must have $\Omega(n^2)$ time labels. The reason is that any foremost realization $\G'$ of an $n\times n$ distance matrix $D$ containing $p$ pairwise distinct entries must have at least $p$ time labels. Indeed, any entry $D_{uv}$ must correspond to some foremost temporal $uv$-path in $\G'$ whose last edge appears at time $D_{uv}$.
\end{proof}

\subparagraph{Remark.} The above lower bound also applies to non-strict foremost realizations (using a similar proof with the same families of matrices and temporal graphs). It also holds if we restrict the problem to sparse prescribed graphs with $O(n)$ edges, or even trees, as long as the star is a possible prescribed graph.

\iflong
\begin{proposition}
\label{prop:lower-bound single edge}
	There exists a family $(D^n)_{n\in \mathbb{N}_{>0}}$ of $n\times n$ distance matrices that are $\Fo$-realizable and such that any realization requires at least one edge with $\Omega(n)$ time labels.
\end{proposition}
\fi
\iflong
\begin{proof}
We obtain the results by deriving a similar result from the following problem.

	\prob{\Reach}
{A directed graph~$G=(V,A)$.}
{Is there a undirected temporal graph~$\mg$ with \emph{strict reachability graph} equal to~$G$, that is, is there for each pair of distinct vertices~$(u,v)$ a strict temporal path from~$u$ to~$v$ in~$\mg$ if and only if $(u,v)$ is an arc of~$G$?}

Erlebach et al.~\cite[Theorem~4]{EMM25} presented a family $(G^n)_{n\in \mathbb{N}_{>0}}$ of directed graphs on~$n$ vertices, for which~\Reach is a yes-instance and where each temporal graph~$\mg$ with reachability graph~$G^n$ has at least one edge with $\Omega(n)$~labels. 
Based on this family of graphs, we define our family of distance matrices.
For each~$n$, let~$\mg^n$ be an arbitrary but fixed temporal graph with reachability graph~$G^n$.
We define~$D^n$ to be the foremost matrix of~$\mg^n$.
Clearly, $D^n$ is $\Fo$-realizable (namely, by~$\mg^n$).
It thus remains to show that each temporal graph~$\mg'$ with foremost matrix~$D^n$ has at least one edge with at least~$\Omega(n)$ labels.
To this end, we show that the reachability graph of~$\mg'$ equals~$G^n$.
This then implies the statement due to the properties of~$G_n$. 
By definition of~$D^n$, $D^n_{u,v} < \infty$ if and only if~$(u,v) \in A$.
Thus, the reachability graph of~$\mg'$ is~$G^n$, as the foremost matrix of~$\mg'$ equals~$D^n$.
\end{proof}
\fi

\medskip
We now present a simple algorithm for computing a $\Fo$-realization that leads to $\Oh(n^4)$ time complexity.
It is based on the observation that any realization $\G$ of a distance matrix $D$ for $\Fo$-path must satisfy the following compatibility property.

\begin{definition}[Edge compatibility]\label{def:compat}
	Given an $n\times n$ distance matrix $D$, a temporal edge $(\{v,w\},\tau)$ is said to be \emph{$\Fo$-edge-compatible} with $D$ if it satisfies:
	$EdgeCompat(D,\{v,w\},\tau) := \forall x\in [n], 
	D_{xv}<\tau\Longrightarrow D_{xw}\leq \tau 
	\text{ and } D_{xw}< \tau \Longrightarrow D_{xv}\le \tau$.
\end{definition}

The intuition behind is the following. If an edge $\{v,w\}$ appears at time $\tau$ in a temporal graph $\G$, then any foremost temporal $xv$-path with arrival time less than $\tau$ can be extended by the temporal edge $(\{v,w\},\tau)$, implying that the foremost arrival time at $w$ is at most $\tau$. Moreover, if $\G$ is a realization of $D$, i.e. $D=\Fo(\G)$, and $D_{xv}<\tau$, then we must have $D_{xw}\le \tau$. By symmetry of edges, we must also have $D_{xw}< \tau \Longrightarrow D_{xv}\le \tau$. 
Note that, given an $n\times n$ distance matrix $D$ and a temporal edge $(\{v,w\},\tau)$, the property $EdgeCompat(D,\{v,w\},\tau)$ can easily be tested in $\Oh(n)$ time.

\smallskip
The algorithm consists in checking that for each entry $\tau=D_{uw}$ of $D$, there exists a vertex $v$ such that $D_{uv}<\tau$ and the temporal edge $(\{v,w\},\tau)$ is $\Fo$-edge-compatible with $D$. If this is the case, such a temporal edge $(\{v,w\},\tau)$ is added (see Algorithm~\ref{alg:foremost}). This condition is indeed necessary as a foremost temporal $uw$-path in a realization of $D$ must end with such an edge. We will show that this condition is also sufficient, and leads to a construction of a realization with at most $n^2$ time labels.

\begin{algorithm}[ht]
	\caption{\FoTGR}\label{alg:foremost}
	\begin{algorithmic}[1]
	\Require A number $n$ and an $n\times n$ distance matrix $D$
	\Ensure YES if there exists a temporal graph $(G,\lambda)$ that is a $\Fo$-realization of $D$, NO otherwise
	
	\State Let $C$ be the set of all couples $(u,w)$ such that $u\neq w$ and $D_{uw}<\infty$
	\State Let $\lambda$ be an empty labeling, i.e. $\lambda(\{u,v\})=\emptyset$ for all $u,v\in [n]$ 
	\State Return \textsc{Exists($D,C,\lambda$)}
	
	\smallskip
	
	\Procedure{Exists}{$D,C,\lambda$}
		\Comment{Given an $n\times n$ distance matrix $D$, $C\subset[n]\times[n]$ and a labeling $\lambda$}
		
		\ForAll{$(u,w)\in C$}
				\If{$\exists v\in [n]$ such that $D_{uv}<D_{uw}$ and $EdgeCompat(D,\{v,w\},D_{uw})$}\label{lin:test-inf}
					\State Add $D_{uw}$ to $\lambda(\{v,w\})$
				\Else 
					\State Return NO 
				\EndIf
		\EndFor
		\State Return YES
	\EndProcedure
	
	\end{algorithmic}
	\end{algorithm}
	

The correctness of the algorithm comes with defining how a temporal graph can partially realize a matrix as follows.

\begin{definition}
	A temporal graph $\G$ is said to be \emph{$\Fo$-compatible} with $D$ if $D\le \Fo(\G)$, i.e. $D_{uv}\le \Fo(\G)_{uv}$ for all $u,v\in [n]$, and all temporal edges of $\G$ are $\Fo$-edge-compatible with $D$.
\end{definition}

First note that any $\Fo$-realization of $D$ must be $\Fo$-compatible with $D$.

\begin{lemma}\label{lem:compat}
	If $\G$ is a $\Fo$-realization of $D$, then $\G$ is $\Fo$-compatible with $D$.
\end{lemma}

\begin{proof}
	First, we clearly have $D\le\Fo(\G)$ since $\G$ is a $\Fo$-realization of $D$.
	Second, suppose for the sake of contradiction that some temporal edge $(\{v,w\},\tau)$ is not $\Fo$-edge-compatible with $D$. That is, without loss of generality, there exists a vertex $x$ such that $D_{xv}<\tau$ and $D_{xw}>\tau$.
	Any foremost temporal $xv$-path in $\G$ must arrive in $v$ at time $D_{xv}$.
	However, this temporal path can be extended with $(\{v,w\},\tau)$, yielding a temporal $xw$-walk 
	arriving at time $\tau<D_{xw}$, contradicting the fact that $\G$ is a $\Fo$-realization of~$D$. 
\end{proof}

Note also that an $n$-vertex empty temporal graph, i.e. without any time label,  is always $D$ compatible as it does not contain any temporal edge. Moreover,  $\Fo$-compatibility is preserved by addition of a $\Fo$-edge-compatible temporal edge as stated below.

\begin{lemma}\label{lem:compat-add}
	If a temporal graph $\G=(G,\lambda)$ is $\Fo$-compatible with $D$ and a temporal edge $(\{v,w\},\tau)$ is $\Fo$-edge-compatible with $D$, then the temporal graph $\G'$ obtained from $\G$ by adding label $\tau$ to edge $\{v,w\}$ is also $\Fo$-compatible.
\end{lemma}

\begin{proof}
	We just need to prove $D\le \Fo(\G')$. Suppose for the sake of contradiction that there exists $x\not=y$ such that $D_{xy}>\Fo(\G')_{xy}$, i.e., $D_{xy}$ is greater than the arrival time of a foremost temporal $xy$-path $P$ in $\G'$. Consider the first temporal edge $(\{z,z'\},\sigma')$ of $P$ such that $P$ arrives in $z$ at time $\sigma\ge D_{xz}$ but arrives in $z'$ before $D_{xz'}$, i.e. $\sigma'<D_{xz'}$. Such an edge must exist since $P$ arrives in $y$ before $D_{xy}$ and $D_{xx}=0$. The (strict) temporal path definition implies $\sigma<\sigma'$ which thus yields $D_{xz}<\sigma'$. As $D_{xz'}>\sigma'$, $(\{z,z'\},\sigma')$ cannot be $\Fo$-edge-compatible with $D$.
	This contradicts either 
	the $\Fo$-edge-compatibility of $(\{v,w\},\tau)$ if $(\{z,z'\},\sigma')=(\{v,w\},\tau)$, or the $\Fo$-compatibility of $\G$ otherwise.
\end{proof}

\begin{lemma}\label{lem:compat-real}
	If $\G$ is a $\Fo$-realization of $D$, then for any entry $\tau=D_{uw}$ with $u\not= w$, there exists a vertex $v$ such that $D_{uv}<\tau$ and $(\{v,w\},\tau)$ is $\Fo$-edge-compatible with~$D$.
\end{lemma}

\begin{proof}
It suffices to consider the last temporal edge $(\{v,w\},\tau)$ of a foremost temporal $uw$-path in an $\Fo$-realization $\G$ of $D$. It must satisfy $\tau=D_{uw}$ since $\G$ is a realization of $D$. Since it is a strict temporal path, it arrives in $v$ before $\tau$, implying $D_{uv}<\tau$. Moreover, $(\{v,w\},\tau)$ is $\Fo$-edge-compatible with $D$ by Lemma~\ref{lem:compat}.
\end{proof}

\begin{proof}[Proof of Theorem~\ref{th:foremost}]
	If there exists a $\Fo$-realization of $D$, the algorithm must find a suitable temporal edge $(\{v,w\},D_{uw})$ for each pair $(u,w)$ by Lemma~\ref{lem:compat-real}. It thus returns NO, only when no such realization exists.
	Let $K_n$ denote the complete graph with vertex set $[n]$.
	Lemma~\ref{lem:compat-add} implies that Algorithm~\ref{alg:foremost} preserves the invariant that $(K_n,\lambda)$ is $\Fo$-compatible with $D$. 
	If the algorithm returns YES, the constructed temporal graph $\G=(K_n,\lambda)$ is thus $\Fo$-compatible with $D$, implying $D\le\Fo(\G)$. We now prove that we indeed must have $D=\Fo(\G)$. Suppose for the sake of contradiction that there are pairs $(u,w)$ satisfying $D_{uw}<\Fo(\G)_{uw}$. Consider such a pair $(u,w)$ such that $D_{uw}$ is minimum. When this pair was considered, the algorithm added to $\G$ a temporal edge $(\{v,w\}, D_{uw})$ satisfying $D_{uv}<D_{uw}$.
	By the choice of $(u,w)$, we have $D_{uv}=\Fo(\G)_{uv}$. Now, if we extend a foremost temporal $uv$-path in $\G$ with $(\{v,w\}, D_{uw})$, we obtain a temporal $uw$-walk arriving at time $D_{uw}$ in contradiction with $D_{uw}<\Fo(\G)_{uw}$. This concludes the proof of correctness of Algorithm~\ref{alg:foremost}.\lvtodo{Should we put this part of the proof in a Lemma: $D$ is $\Fo$-realizable iff $\forall (u,w), \exists v$ s.t. $D_{uv}<D_{uw}$ and $EdgeCompat(D,\{v,w\},D_{uw})$. This would be usefull when proving the FPT algorithm.}

	Its time complexity is clearly $\Oh(n^4)$ as for each of the $n^2$ pairs $(u,w)$, we consider at most $n$ vertices $v$ and the test for the $\Fo$-edge-compatibility of $(\{v,w\},D_{uw})$ takes $\Oh(n)$ time. 
	To obtain $\Oh(n^3\log n)$, we use an interval tree data-structure (see, e.g., ~\cite{CormenLRS2009}). It can store $n$ intervals and querying whether a value $\tau$ is in one of these intervals can be answered in $O(\log n)$ time. It uses $O(n)$ space and can be constructed in $O(n\log n)$ time.
	To benefit from such a data-structure, we consider all pairs $(u,w)$ with fixed $w$ consecutively. Before processing them, we compute for each vertex $v$ two interval trees $T_v$ and $T_v'$ where $T_v$ (resp. $T_v'$) contains the $n$ intervals $[D_{xv}+1,D_{xw}-1]$ (resp. $[D_{xw}+1,D_{xv}-1]$) for $x\in [n]$, ignoring empty intervals. The $\Fo$-edge-compatibility of a temporal edge $(\{v,w\},\tau)$ can then be tested in $O(\log n)$ time by checking that neither $T_v$ nor $T_v'$ has an interval containing $\tau$. The reason is that any vertex $x$ violating $D_{xv}< \tau \Longrightarrow D_{xw}\le \tau$ satisfies $D_{xv}<\tau<D_{xw}$ in which case $\tau$ belongs to the interval of $T_v$ associated to $x$. Analogously, any vertex $x$ violating $D_{xw}< \tau \Longrightarrow D_{xv}\le \tau$ is associated to an interval of $T_v'$ that contains $\tau$.
	Constructing the interval trees takes $O(n^2\log n)$ time while processing each $(u,w)$ pair now takes $O(n\log n)$ time as it mainly consists in $n$ $\Fo$-edge-compatibility tests. The overall complexity is thus $O(n^3\log n)$ time using $O(n^2)$ space.  

	Note that the algorithm adds a time label to an edge $\{v,w\}$ at most once for each vertex $u$, when considering either $(u,w)$ or $(u,v)$, depending on whether $D_{uv}<D_{uw}$ or $D_{uw}<D_{uv}$. The $\Fo$-realization computed by Algorithm~\ref{alg:foremost} thus has at most $n$ time labels per edge, and at most $n^2$ time labels in total.
\end{proof}



\paragraph*{Non-strict foremost paths}

We have a similar result for the non-strict case.
\iflong
\begin{theorem}
\else
\begin{theorem}[$\star$]
\fi
\label{th:NS-foremost}
	\NSFoTGR can be solved in $\Oh(n^3\log n)$ time and $\Oh(n^2)$ space. Furthermore, if dealing with a realizable instance, a realization with at most $n^2$ time labels can be computed with same complexity.
\end{theorem}

A slight modification of Algorithm~\ref{alg:foremost} suffices, replacing $EdgeCompat(D,\{v,w\},\tau)$ with:
$$\NSECompat(D,\{v,w\},\tau) :=\forall x\in [n], D_{xv}\le \tau\Longleftrightarrow D_{xw}\leq \tau.$$

\iflong
Following the above definition, we say that a temporal edge $(\{v,w\},\tau)$ is \emph{$\NSFo$-edge-compatible} with $D$ if it satisfies $\NSECompat(D,\{v,w\},\tau)$. Similarly, a temporal graph $\G$ is said to be \emph{$\NSFo$-compatible} with $D$ if $D\le \NSFo(\G)$ and all temporal edges of $\G$ are $\NSFo$-edge-compatible. The following variations of Lemmas~\ref{lem:compat} and~\ref{lem:compat-add} have almost identical proofs which are omitted.

\begin{lemma}\label{lem:NS-compat}
	If $\G$ is a $\NSFo$-realization of $D$, then $\G$ is $\NSFo$-compatible with $D$.
\end{lemma}

\begin{lemma}\label{lem:NS-compat-add}
	If a temporal graph $\G=(G,\lambda)$ is $\NSFo$-compatible with $D$ and a temporal edge $(\{v,w\},\tau)$ is $\NSFo$-edge-compatible with $D$, then the temporal graph $\G'$ obtained from $\G$ by adding label $\tau$ to edge $\{v,w\}$ is also $\NSFo$-compatible.
\end{lemma}


The main difference with the strict setting concerns Lemma~\ref{lem:compat-real} which has the following variation.

\begin{lemma}\label{lem:NS-compat-real}
	If $\G$ is a $\NSFo$-realization of $D$, then for any entry $\tau=D_{uw}$ with $u\not= w$, there exists a vertex $v$ such that $D_{uv}<\tau$ and $(\{v,w\},\tau)$ is $\NSFo$-edge-compatible with $D$.
\end{lemma}

\begin{proof}
It suffices to consider a non-strict foremost temporal $uw$-path $P$. 
The path~$P$ must end with one or more temporal edges with time label $\tau=D_{uw}$ since $\G$ is a $\NSFo$-realization of $D$. Consider the first such temporal edge $(\{v,w'\}, \tau)$.
This choice implies that $P$ arrives in $v$ before $\tau$, implying $D_{uv}<\tau$.
We now show that $(\{v,w\}, \tau)$ is $\NSFo$-edge-compatible with $D$.
For the sake of contradiction, suppose without loss of generality that there exists a vertex $x$ such that $D_{xv}\le \tau$ and $D_{xw}>\tau$. Consider a non-strict foremost temporal $xv$-path $P'$ in $\G$. 
The path~$P'$ must arrive in $v$ at time $D_{xv}$ since $\G$ is a $\NSFo$-realization of $D$. As $D_{xv}\le \tau$, it can be extended by the suffix of $P$ consisting of temporal edges appearing at time $\tau$. This leads to a non-strict temporal $xw$-walk 
arriving at time $\tau$. This implies $D_{xw}=\NSFo(\G)_{xw}\le \tau$, which contradicts $D_{xw}>\tau$.
\end{proof}

Note that the temporal edge $(\{v,w\},\tau)$ is not necessarily in $\G$, it can be a shortcut for a non-strict temporal path where all edges are traversed at time $\tau$.
In comparison, the temporal edge proposed in the proof of Lemma~\ref{lem:compat-real} is in the considered realization.

The rest of the proof of Theorem~\ref{th:NS-foremost} is almost identical to that of Theorem~\ref{th:foremost} and is omitted.
\else
The change of $<$ for $\le$ accounts for considering non-strict temporal paths rather than strict ones. This modification requires a different version of Lemma~\ref{lem:compat-real} with a significantly different proof where 
a non-strict temporal $vw$-path whose edges are traversed at same time $\tau$ in a realization is replaced by a single temporal edge $(\{v,w\},\tau)$ (see Lemma~3.12 
in the appendix).
\fi

\iflong
\paragraph*{Foremost paths in a periodic temporal graph}

We now consider the variant where the realization is required to be periodic. Recall that \textsc{Periodic \FoTGR} is given a period $\Delta$ in addition to $D$, and asks whether there exists a $\Delta$-periodic temporal graph $\G$ such that $\Fo(G)=D$.
We obtain similarly the following.

\begin{theorem}\label{th:Period-foremost}
	\textsc{Periodic \FoTGR} can be solved in $\Oh(n^3\log n)$ time and $\Oh(n^2)$ space. Furthermore, if dealing with a realizable instance, a realization with at most $n^2$ time labels within any period can be computed with same complexity.
\end{theorem}

The proof is almost identical to that of Theorem~\ref{th:foremost} using the following notion of edge-compatibility:
$$
PeriodEdgeCompat(D,\{v,w\},\tau) := \forall i\in\mathbb{Z}, EdgeCompat(D,\{v,w\},\tau+i\Delta).
$$
%
Note that the above condition is restrictive only for $\tau+i\Delta > 0$ as $EdgeCompat(D,\{v,w\},\tau')$ is always satisfied for $\tau'\le 0$.

The correctness of the resulting variant of Algorithm~\ref{alg:foremost} mainly comes from the following observation. Any $\Delta$-periodic realization must satisfy the above condition for each temporal edge $(\{v,w\},\tau)$ as each periodic appearance of $\{v,w\}$ at $\tau+i\Delta$ must be $\Fo$-edge-compatible by Lemma~\ref{lem:compat}.

Note that the above condition can again be performed in $O(n)$ time by checking both $PeriodECDir(D,v,w,\tau)$ and $PeriodECDir(D,w,v,\tau)$ where $PeriodECDir(D,v,w,\tau,\Delta) := \forall x\in [n], \forall i\in\mathbb{Z}, D_{xv} < i\Delta+\tau  \Longrightarrow D_{xw}\le i\Delta +\tau$. In the latter test, it is sufficient to test for each $x\in [n]$ only the first index $i=\lfloor \frac{D_{xv}-\tau}{\Delta} + 1 \rfloor$ satisfying $D_{xv} < i\Delta+\tau $. The complexity thus remains the same.

\paragraph*{Foremost paths with prescribed graph}

We now consider the variant in which a prescribed graph
$G_p=([n],E_p)$ is additionally given as input, and the realization is required to have a subgraph of $G_p$ as underlying graph. We let $N_p(v)=\{w:\exists \{v,w\}\in E_p\}$ denote the set of neighbors of any vertex $v\in [n]$ in $G_p$.
Recall that the problem is formally defined as follows.
\prob{Prescribed \FoTGR}
{A number $n$, an $n\times n$ distance matrix $D$, a prescribed graph $G_p=([n],E_p)$.}
{Is there a temporal graph $\mg=(G,\lambda)$ with $n$ vertices such that $\Fo(\G)$ equals $D$ and $G$ is a subgraph of $G_p$?}

We obtain similarly the following.

\begin{theorem}\label{th:Pr-foremost}
	\textsc{Prescribed \FoTGR} can be solved in $\Oh(n^3\log n)$ time and $\Oh(n^2)$ space. Furthermore, if dealing with a realizable instance, a realization with at most $n^2$ time labels can be computed with same complexity.
\end{theorem}

Again, the proof is almost identical to that of Theorem~\ref{th:foremost} using the following notion of edge-compatibility that takes into account that edges must be in the prescribed graph:
$$
PrescribedEdgeCompat(D,\{v,w\},\tau) := EdgeCompat(D,\{v,w\},\tau) \text{ and } \{u,w\}\in E_p.
$$

\else
\paragraph*{Periodic temporal graph and prescribed graph}

It is straightforward to generalize \Cref{th:foremost} to \textsc{Periodic \FoTGR}  and \textsc{Prescribed \FoTGR} with appropriate definitions of edge compatibility (see Therems~3.13 and~3.14 in the appendix).
\fi

\paragraph*{Non-strict foremost paths with a prescribed graph}

\iflong
We now consider the non-strict version of the prescribed foremost realization problem.
\else
In this setting, a prescribed graph
$G_p=([n],E_p)$ is additionally given as input, and the realization is required to have a subgraph of $G_p$ as underlying graph. We let $N_p(v)=\{w:\exists \{v,w\}\in E_p\}$ denote the set of neighbors of any vertex $v\in [n]$ in $G_p$.
\fi
%
\iflong
Recall that in the adaptation of Algorithm~\ref{alg:foremost} for the non-strict case, we were using the fact that, in a $\NSFo$-realization $\mg$ of $D$, a non-strict temporal path where all edges are traversed at time $\tau$ from a node $v$ to a node $w$ can be replaced by a single temporal edge $(\{v,w\},\tau)$ (see Lemma~\ref{lem:NS-compat-real}). The main difficulty here is that the shortcut temporal edge might not be part of the prescribed graph. 
\else
The main idea is again to add temporal edges $(\{v,w\},\tau)$ that satisfy $NSEdgeCompat(D,\{v,w\},\tau)$ and such that $\{v,w\}$ is present in the prescribed graph. But conversely to the non-strict setting considered in \Cref{th:NS-foremost}, it is not possible to replace an instantaneous temporal $vw$-path (whose edges are traversed at the same time) in a realization by a single appearance of $\{v,w\}$ as this edge might not be in the prescribed graph.  
\fi
Indeed, the condition $D_{uv} < D_{uw}$ at Line~\ref{lin:test-inf} of \Cref{alg:foremost} now becomes problematic as the prescribed graph may impose to add a temporal edge $(\{v,w\},D_{uw})$ such that $D_{uv}=D_{uw}$ to fulfill an entry $D_{uw}$. Moreover, the order in which we can fulfill entries $D_{uw_1}=\cdots=D_{uw_p}$ in this manner may depend on the prescribed graph. A naive solution would be to let each edge $\{v,w\}\in E_p$ appear at all times $\tau\notin\{0,\infty\}$ appearing in $D$ that satisfy $NSEdgeCompat(D,\{v,w\},\tau)$. It would then suffice to check if the resulting temporal graph is an $\NSFo$-realization of $D$. However, this would result in a poor complexity and possibly $\Theta(n^4)$ time labels overall.
We can still solve the problem with a better complexity and a tight number of time labels as stated below.

\iflong
\begin{theorem}
\else
\begin{theorem}[$\star$]
\fi
\label{th:Pr-NS-foremost}
	\textsc{Prescribed \NSFoTGR} can be solved in $\Oh(n^2m)$ time and $\Oh(n^2)$ space, where $m$ is the number of edges of the prescribed graph. Furthermore, if dealing with a realizable instance, a realization with at most $n^2$ time labels can be computed with same complexity.
\end{theorem}

\iflong
\begin{algorithm}[ht]
	\caption{Prescribed \NSFoTGR}\label{alg:pres-NS-foremost}
	\begin{algorithmic}[1]
	\Require A number $n$, an $n\times n$ distance matrix $D$, a prescribed graph $G_p=([n],E_p)$
	\Ensure YES if there exists a temporal graph $(G_p,\lambda)$ that is a $\NSFo$-realization of $D$
	, NO otherwise
	
	\State Let $d_1,\ldots,d_p$ be 
	an enumeration of the set $\{D_{uw} : u\not=w \text{ and } D_{uw}\not=\infty \}$
	\State Let $\lambda$ be an empty labeling, i.e. $\lambda(\{u,v\})=\emptyset$ for all $\{u,v\}\in E_p$ 	
	\State $Mark_{uw}:=False$ for all $u,w\in [n]$
	\ForAll{$i=1,\dots,p$}
		\State Let $C_i$ be the set of all couples $(u,w)$ such that $D_{uw}=d_i$
		\ForAll{$(u,w)\in C_i$} \label{lin:uw-main}
			\If{$\neg$ $Mark_{uw}$}
				\If{$\exists v\in N_p(w)$ such that $D_{uv}<D_{uw}$ and $\NSECompat(D,\{v,w\},d_i)$}
					\State $Mark_{uw}:=True$ \label{lin:mark1}
					\State Add $d_i$ to $\lambda(\{v,w\})$
					\State $BFS(u,w,d_i)$
				\EndIf
			\EndIf
		\EndFor
		\If{$\exists (u,w)\in C_i$ such that $\neg Mark_{uw}$}
			\State Return NO
		\EndIf
	\EndFor	
	\State Return YES

	\smallskip
	
	\Procedure{BFS}{$u,w,d_i$}
		\Comment{Given two nodes $u,w$ and the time $d_i=D_{uw}$}
		
		\State Add $w$ to queue $Q$
		\While{$Q\neq \emptyset$}
			\State Let $v:=pop(Q)$ \label{lin:pop}
			\ForAll{$w'\in N_p(v)$ such that $D_{uw'}=d_i$ and $\neg Mark_{uw'}$}
				\If{$\NSECompat(D,\{v,w'\},d_i)$}
					\State $Mark_{uw'}:=True$ \label{lin:mark2}
					\State Add $d_i$ to $\lambda(\{v,w'\})$
					\State Add $w'$ to $Q$
				\EndIf
			\EndFor		
		\EndWhile
		
	\EndProcedure
	
	\end{algorithmic}
\end{algorithm}
\else
\fi

The result is a consequence of 
\iflong
Algorithm~\ref{alg:pres-NS-foremost} which scan the set $\{d_1,\ldots,d_p\}$ of entries of $D$ excluding $0$ and $\infty$ (in any order).
\else
the following algorithm. It scans the set $\{d_1,\ldots,d_p\}$ of entries of $D$ excluding $0$ and $\infty$ (in any order). For each $d_i$, it checks the set $C_i$ of all couples $(u,w)$ such that $D_{uw}=d_i$. If there exists a vertex $v\in N_p(w)$ such that $D_{uv}<D_{uw}$ and $\NSECompat(D,\{v,w\},d_i)$ is satisfied, it adds label $D_{uw}$ to $\{v,w\}$, similarly to \Cref{alg:foremost}. In addition, it starts a BFS like procedure to find other couples $(u,w')\in C_i$ that can be reached at time $d_i$ through $w$.
See Algorithm~2 
in the appendix for more details.
\fi

\iflong
 The set $C_i$ of all couples $(u,w)$ with same value $D_{uw}=d_i$ is processed as a whole. As with the prescribed non-strict variant of Algorithm~\ref{alg:foremost}, we add label $d_i$ to an edge $\{v,w\}$ if we find a neighbor $v$ of $w$ in $G_p$ such that the edge with time label $D_{uw}$ is edge-compatible with $D$ according to $\NSECompat$ and $u$ can reach $v$ before $w$ according to $D$, i.e., $D_{uv} < D_{uw}$. However, this is not sufficient as some pairs $(u,w)$ might need to use an edge $\{w',w\}$ at time $d_i$ where $(u,w')$ is another couple of $C_i$. For that purpose, for each couple $(u,w)$ where a compatible edge $\{v,w\}$ with $D_{uv} < D_{uw}$ was found, we start a BFS like search to scan and mark other couples $(u,w')$ that can be reached at time $d_i$ through $w$ based on the following lemma.

\begin{lemma}\label{lem:NS-compat-real-bis}
	If $\G$ is a $\NSFo$-realization of $D$, then for any entry $\tau=D_{uw}$ with $u\not= w$, there exists vertices $v_1,\ldots,v_k$ with $k\ge 2$ and $v_k=w$ such that $D_{uv_1}<\tau$ and such that the temporal edges $(\{v_i,v_{i+1}\},\tau)$ for $i\in [k-1]$ are all $\NSFo$-edge-compatible with $D$ and are all in $\G$. 
\end{lemma}

\begin{proof}
	It suffices to consider a non-strict foremost temporal $uw$-path in $\G$ and its longest suffix $(\{v_1,v_2\},\tau),\dots,(\{v_{k-1},w\},\tau)$ of temporal edges with time label $\tau=D_{uw}$. All these temporal edges are $\NSFo$-edge-compatible by Lemma~\ref{lem:NS-compat}.
\end{proof}

\begin{proof}
	We now prove the correctness of Algorithm~\ref{alg:pres-NS-foremost}.
	First, suppose that there exists a $\NSFo$-realization $\G$ such that its underlying graph is a subgraph of $G_p$. We prove that all couples $(u,w)$ satisfying $0<D_{uw}<\infty$ are marked. Consider any couple $(u,w)\in C_i$ and let us prove that it will be marked, i.e. $Mark_{uw}$ will be set to $True$. 
	Applying Lemma~\ref{lem:NS-compat-real-bis} with $\tau=d_i=D_{uw}$, there exists vertices $v_1,\ldots,v_k$ with $v_k=w$ and $D_{uv_1}<d_i$ such that each temporal edge $(\{v_j,v_{j+1}\},d_i)$ is in $\G$ and is $\NSFo$-edge-compatible with $D$. We have the following properties:
	\begin{itemize} 
		\item[(i)] All edges $\{v_j,v_{j+1}\}$ for $j\in [k-1]$ are in $E_p$ since the underlying graph of $\G$ is a subgraph of $G_p$.
		\item[(ii)] $D_{uv_j}\le d_i$ for all $j\in \{2,\ldots,k\}$ since $D_{uv_1}=\NSFo(\G)_{uv_1}< d_i$ and a non-strict foremost temporal $uv_1$-path can be extended by $(\{v_j,v_{j+1}\},d_i)$ for $j\in [k-1]$, leading to $d_i\ge \NSFo(\G)_{uv_j}=D_{uv_j}$ for $1<j\le k$.
	\end{itemize}
	Let $j$ be the last index in $[k]$ such that $D_{uv_j}<d_i$.
	Since $D_{uw}=D_{uv_k}=d_i$, we have $j<k$ and $D_{uv_{j+1}}=d_i$.
	As $(\{v_j,v_{j+1}\},d_i)$ is $\NSFo$-edge-compatible with $D$, the algorithm must mark $(u,v_{j+1})$ at Line~\ref{lin:mark1} if it has not been already marked during a call to $\text{\sc BFS}$ at  Line~\ref{lin:mark2}.  
	Suppose for the sake of contradiction that $(u,w)$ is not marked. Let $j'$ be the last index in $j+1,\ldots,k-1$ which is marked.
	This happens either at Line~\ref{lin:mark1} or Line~\ref{lin:mark2}. In both cases, $v_{j'}$ is enqueued in $Q$ during a call to $\text{\sc BFS}$, and later popped at Line~\ref{lin:pop}. The algorithm then scans $w'=v_{j'+1}$ and must mark it at Line~\ref{lin:mark2} since $D_{uv_{j'+1}}=d_i$ by (ii) and the choice of $j$, and $(\{v_{j'},v_{j'+1}\},d_i)$ is $\NSFo$-edge-compatible with $D$, in contradiction with $Mark_{uv_{j'+1}}=False$ by the choice of $j'$.
	We thus conclude that all couples $(u,w)$ with $0<D_{uw}<\infty$ are marked, and the algorithm thus returns YES.

	Now, note the following invariant: $(G_p,\lambda)$ is $\NSFo$-compatible at any point of the execution of the algorithm. This is a direct consequence of Lemma~\ref{lem:NS-compat-add}. We thus have $D\le\NSFo((G_p,\lambda))$. We prove that we indeed have $D=\NSFo((G_p,\lambda))$. Suppose for the sake of contradiction that there exists a couple $(u,w)$ satisfying $D_{uw}<\NSFo((G_p,\lambda))_{uw}$. Consider such a couple $(u,w)$ such that $D_{uw}$ is minimum. Consider the index $i$ such that $D_{uw}=d_i$. Moreover, consider the first couple $(u,w')$ which is marked among the couples $(u,w)\in C_i$ satisfying $D_{uw}<\NSFo((G_p,\lambda))_{uw}$. It must have been marked when finding a $\NSFo$-edge-compatible temporal edge $(\{v,w'\},d_i)$ which has been added to $(G_p,\lambda)$. If this happens at Line~\ref{lin:mark1}, we have $D_{uv}<D_{uw'}$ by the choice of $v$, implying $D_{uv}=\NSFo((G_p,\lambda))_{uv}$ by minimality of $D_{uw}$. Otherwise, it happens at Line~\ref{lin:mark2}, and $(u,v)$ was already marked since $v$ was popped at Line~\ref{lin:pop}. This implies $D_{uv}=\NSFo((G_p,\lambda))_{uv}$ by the choice of $w'$. In both cases we have $D_{uv}=\NSFo((G_p,\lambda))_{uv}$, and extending a non-strict foremost temporal $uv$-path in $(G_p,\lambda)$ by the temporal edge $(\{v,w'\},d_i)$ leads to a non-strict temporal $uw'$-walk arriving at time $d_i=D_{uw'}$ in contradiction with $D_{uw'}<\NSFo((G_p,\lambda))_{uw'}$.
	This concludes the proof of correctness of Algorithm~\ref{alg:pres-NS-foremost}.

	The complexity of the algorithm comes from the following observation.
	For each couple $(u,w)$, we consider the neighbors $v$ of $w$ at most twice: once when considering $(u,w)$ at Line~\ref{lin:uw-main} in the main loop, and once when $w$ is popped during a BFS call. The reason is that a node $w$ is enqueued only if $(u,w)$ has just been marked and $(u,w)$ was not marked before. For each such neighbor $x$, we perform a test of edge compatibility in $O(n)$. The overall complexity is thus
	$\Oh(\sum_{uw} |N_p(w)| n)=\Oh(n^2\sum_w |N_p(w)|)=\Oh(n^2m)$.

	Finally, a time label is added to an edge $\{v,w\}$ at most once for each node $u$. Indeed, if $D_{uv}<D_{uw}$ (resp. $D_{uw}<D_{uv}$), it can only be added when fulfilling $(u,w)$ (resp. $(u,v)$). Ohterwise, if $D_{uv}=D_{uw}$, it can only be added when one of $v,w$ is marked and not the other, i.e., when fulfilling either $(u,w)$ or $(u,v)$. 
	The constructed realization $(G_p,\lambda)$ thus has at most $n^2$ time labels in total.
\end{proof}
\else
\fi

\subsection{Limits of polynomial-time algorithms for \EARLY}
In this section, we show several additional requirements on instances of~\EARLY for which the problem becomes NP-hard.

\iflong
\paragraph*{Foremost paths with a single time-label per edge}
\fi
\iflong
The problem becomes hard when we restrict the number of allowed labels per edge to one. 
\else
\fi
\iflong
\begin{theorem}
\else
\begin{theorem}[$\star$]
\fi\label{foremost simple}
		\EARLY is NP-hard when allowing at most one label per edge. 
	\end{theorem}
	\iflong
\begin{proof}
We reduce from~\SAT.

\prob{\SAT}
{A formula~$F$ in CNF.}
{Is~$F$ satisfiable?}

Let~$F$ be an instance of~\SAT where no clause contains the same variable both positively and negatively.
Under these restrictions, \SAT is NP-hard~\cite{K72}.

\textbf{Construction.}
Let~$X$ be the variable set of~$F$ and let~$C$ denote the clauses of~$F$.
To obtain an instance~$D$ of~\EARLY, we first define the vertex set over which the matrix~$D$ is defined.
We set~$V := \{x,\overline{x}\mid x\in X\} \cup C \cup \{\top,\bot\}$.

We now describe the entries of~$D$.
See~\Cref{fig early simple} for an illustration.
Let~$c\in C$. 
We set~$D_{c,\top} := 3$.
Furthermore, for each variable~$x$ that occurs positively in~$c$, we set~$D_{c,x} := D_{x,c} := 2$ and~$D_{\overline{x},c} := 2$.
Similarly, for each variable~$x$ that occurs negatively in~$c$, we set~$D_{c,\overline{x}} := D_{\overline{x},c} := 2$ and~$D_{x,c} := 2$.
For each variable~$x$, we set~$D_{x,\overline{x}} := D_{\overline{x},x} := 1$, and for each literal~$\ell$, we set~$D_{\ell,\top} := D_{\top,\ell} := D_{\ell,\bot} := D_{\bot,\ell} := 2$.
Furthermore, for any two literals~$\ell$ and~$\ell'$ that belong to different variables, we set~$D_{\ell,\ell'} := D_{\ell',\ell} := 2$. We set $D_{\top,\bot}:=1$.
All other entries of~$D$ are set to~$4$.
Note that this includes the entry~$D_{\top,c}$ for each clause~$c$.  
This completes the definition of~$D$.

\textbf{Intuition.}
The difficulty of the constructed instance comes from the realizability of the entries from clause vertices to~$\top$.
Each clause~$c$ can only reach literals that are contained in~$c$ in time~$2$, but has to reach~$\top$ by time~$3$.
This implies that for some literal~$\ell$ in~$c$, the edge~$\{\ell,\top\}$ needs to receive label~$3$, as the direct edge~$\{c,\top\}$ cannot receive a label smaller than~$4$ due to~$D_{\top,c}=4$.
Moreover, for each variable~$x$, the entries of~$D$ between~$\{x,\ol\}$ and~$\{\top,\bot\}$ ensure that at least one of~$\{x,\top\}$ or~$\{\overline{x},\top\}$ has to receive label~$2$.
As we require to have at most one label per edge, note that the edge~$\{\ell,\top\}$ can receive label~$3$ to fulfill $D_{c,\top}$ only if it did not already receive label~$2$.
Hence, we can derive a truth assignment for the edges between each variable gadget and the vertex~$\top$ (by checking which of the edges~$\{x,\top\}$ or~$\{\ol,\top\}$ does not have label 2), which we show to be a satisfying assignment of the formula in each realization of~$D$.

\textbf{Correctness.}
We now show that~$F$ is satisfiable if and only if~$D$ is realizable with a single label per edge.

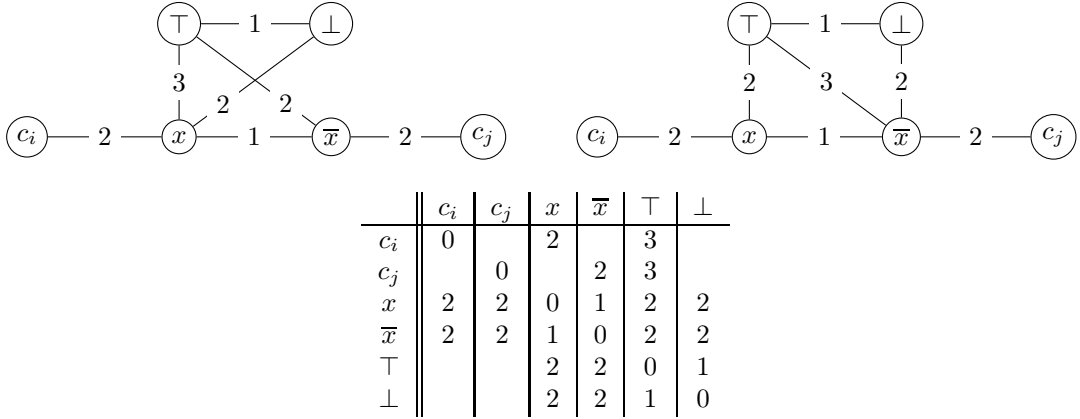
\begin{figure}
\centering
\begin{tikzpicture}

\tikzstyle{k}=[circle,fill=white,draw=black,minimum size=10pt,inner sep=2pt]

\node[k] (l) at (0,0) {$c_i$};
\node[k] (p) at (2,0) {$x$};
\node[k] (n) at (4,0) {$\overline{x}$};
\node[k] (r) at (6,0) {$c_j$};
\node[k] (v) at (2,1.5) {$\top$};
\node[k] (b) at (4,1.5) {$\bot$};

\draw[-] (l) -- (p) node [midway, fill=white] {$2$};
\draw[-] (r) -- (n) node [midway, fill=white] {$2$};
\draw[-] (n) -- (p) node [midway, fill=white] {$1$};

\draw[-] (b) -- (v) node [midway, fill=white] {$1$};

\draw[-] (p) -- (b) node [near start, fill=white] {$2$};
\draw[-] (n) -- (v) node [near start, fill=white] {$2$};

\draw[-] (p) -- (v) node [midway, fill=white] {$3$};

\begin{scope} [xshift=7.5cm]

\node[k] (l) at (0,0) {$c_i$};
\node[k] (p) at (2,0) {$x$};
\node[k] (n) at (4,0) {$\overline{x}$};
\node[k] (r) at (6,0) {$c_j$};
\node[k] (v) at (2,1.5) {$\top$};
\node[k] (b) at (4,1.5) {$\bot$};

\draw[-] (l) -- (p) node [midway, fill=white] {$2$};
\draw[-] (r) -- (n) node [midway, fill=white] {$2$};
\draw[-] (n) -- (p) node [midway, fill=white] {$1$};

\draw[-] (b) -- (v) node [midway, fill=white] {$1$};
\draw[-] (n) -- (b) node [midway, fill=white] {$2$};
\draw[-] (p) -- (v) node [midway, fill=white] {$2$};
\draw[-] (n) -- (v) node [midway, fill=white] {$3$};

\end{scope}
\end{tikzpicture}
\vspace{.4cm}

\begin{tabular}{c||c|c|c|c|c|c}
 & $c_i$ & $c_j$ & $x$ & $\overline{x}$ & $\top$ & $\bot$ \\\hline
$c_i$ & 0 &   & 2 &  & 3 &   \\
$c_j$ &   & 0 &  & 2 & 3 &   \\
$x$ & 2 & 2 & 0 & 1 & 2 & 2  \\
$\overline{x}$ & 2 & 2 & 1 & 0 & 2 & 2  \\
$\top$ &   &   & 2  & 2  & 0 & 1 \\
$\bot$ &   &   & 2 & 2 & 1 & 0
\end{tabular}
\caption{The variable gadget from the reduction of~\Cref{foremost simple} with a true assignment on the left side and a false assignment on the right side. 
Here, the clause~$c_i$ contains the literal~$x$ and the clause~$c_j$ contains the literal~$\ol$.
The matrix only shows the entries of value smaller than 4. 
Each clause vertex~$c_i$ aims to reach~$\top$ by time~$3$, but for each variable~$x$, one of the edges~$\{x,\top\}$ or~$\{\ol,\top\}$ has to receive label~$2$.}
\label{fig early simple}
\end{figure}

$(\Leftarrow)$
Let~$\mg=((V,E),\lambda)$ be a realization of~$D$ where each edge receives a single label.
Note that we can assume that~$E = \binom{V}{2}$, as the largest entry of~$D$ is~$4$, that is, each non-edge can receive label~$4$ while preserving the property of being a realization for~$D$.

First, we show that the edges between each variable gadgets and~$\{\top,\bot\}$ encodes a truth assignment.
\begin{claim}\label{only one two simple}
For each variable~$x\in X$, $\lambda(\{x,\top\}) = 2$ or~$\lambda(\{\ol,\top\}) = 2$.
\end{claim}
\begin{claimproof}
Let~$x\in X$.
Since~$\mg$ realizes~$D$ and~$D_{x,\top} = 2$, there is a temporal path~$P$ from~$x$ in~$\mg$ that reaches~$\top$ at time~$2$.
As~$D_{x,\ol}$ and~$D_{\ol,x}$ are the only entries of~$D$ of value~$1$ that include~$x$, the temporal path~$P$ that realizes~$D_{x,\top} = 2$ consists of either (i)~the single edge~$\{x,\top\}$ with label~$2$, or (ii)~the edge~$\{x,\ol\}$ with label~$1$ followed by the edge~$\{\ol,\top\}$ with label~$2$.
\end{claimproof}

Based on this claim, we define a truth assignment of the variables of~$X$ as follows.
For each variable~$x\in X$, we set variable~$x$ to True if and only if~$\lambda(\{x,\top\}) \neq 2$.
We show that this assignment satisfies~$F$.
To this end, we show that each clause is satisfied by the assignment.
Let~$c$ be a clause.
Since~$D_{c,\top} = 3$ and~$\mg$ realizes~$D$, there is a temporal path~$P$ from~$c$ that reaches~$\top$ at time~$3$.
By $D_{\top,c} = 4$, we get $\lambda(\{c,\top\}) > 3$. 
This implies that $P$ has length at least~$2$.
Let~$q$ be the first internal vertex of~$P$.
Since each entry of~$D$ including~$c$ has value at least~$2$, $P$ has length exactly~$2$ since we consider strict temporal paths and~$P$ ends in time~$3$.
That is, $P=(c,q,\top)$, $\lambda(\{c,q\}) = 2$, and~$\lambda(\{q,\top\}) = 3$.
By~$\lambda(\{c,q\}) = 2$, we get that~$D_{c,q} \leq 2$, which implies that~$q$ is a literal that is contained in clause~$c$.
Thus, literal~$q$ is set to True by our truth assignment, since~$\lambda(\{q,\top\}) = 3 \neq 2$.
Hence, clause~$c$ is satisfied by the truth assignment, which implies that~$F$ is satisfied.

$(\Rightarrow)$
Let~$\pi$ be a satisfying truth assignment for~$F$.
We define a labeling~$\lambda$ for the complete graph~$G$ with vertex set~$V$, such that~$\mg:= (G,\lambda)$ realizes~$D$.
We initialize~$\lambda$ by setting~$\lambda(e) := 1$ for each edge~$e\in \{\{\top,\bot\}\}\cup\{\{x,\ol\}\mid x\in X\}$.
This realizes all entries of~$D$ of value~$1$.
Next, we define the set~$E_2$ of edges that receive label~$2$.
For any two literals~$\ell$ and~$\ell'$ that do not belong to the same variable, we add the edge~$\{\ell,\ell'\}$ to~$E_2$.
This implies that all entries of value~$2$ between literal vertices are realized.
For each clause~$c$ and each literal~$\ell$ of~$c$, we add~$\{c,\ell\}$ to~$E_2$.
This realizes all entries of~$D$ of value~$2$ including any clause vertex.
This is due to the fact that~$c$ and literal~$\ell$ can pairwise reach each other at time~$2$ via the direct edge and the negated literal~$\overline{\ell}$ of~$\ell$ can reach~$c$ via the path using edges~$\{\overline{\ell},\ell\}$ at time~$1$ and~$\{\ell,c\}$ at time~$2$.
The only entries of~$D$ of value~$2$ that are not yet realized are all entries between~$\{\top,\bot\}$ and any literal~$\ell$.
Let~$x\in X$. 
If~$x$ is assigned to True by~$\pi$, we add the edges~$\{x,\bot\}$ and~$\{\ol,\top\}$ to~$E_2$.
Otherwise, we add the edges~$\{x,\top\}$ and~$\{\ol,\bot\}$ to~$E_2$.
In both cases, all vertices of~$\{x,\ol\}$ can reach all vertices of~$\{\top,\bot\}$ and vice versa at time~$2$, since there is a perfect matching at time~$2$ between these sets and the edges~$\{x,\ol\}$ and~$\{\top,\bot\}$ receive label~$1$.
Thus, via the edges of~$E_2$ and the edges with label~$1$, all entries of~$D$ of value at most~$2$ are realized.
Next, we define the edges~$E_3$ that we need to realize all entries of value~$3$ in~$D$.
Note that these are exactly the entries~$D_{c,\top}$ for each clause~$c$.
If~$x$ is assigned to True by~$\pi$, we add the edges~$\{x,\top\}$ to~$E_3$ and otherwise, we add the edges~$\{\ol,\top\}$ to~$E_3$.
By definition of~$E_2$, these edges have not received any label yet.
Moreover, since~$\pi$ is a satisfying assignment, these edges realize all entries of value~$3$, since for each clause~$c$, there is a literal~$\ell$ of~$c$ that is assigned True by~$\pi$.
That is, $\{c,\ell\}\in E_2$ and~$\{\ell,\top\}\in E_3$, which implies that there is a path from~$c$ to~$\top$ that arrives at time~$3$.
Finally, all other edges receive label~$4$.
This clearly realizes all remaining entries of~$D$.
Moreover, as $E_1,E_2$, and~$E_3$ are pairwise disjoint, each edge receives exactly one time label.
Hence, $\mg$ realizes~$D$.
\end{proof}
\fi

\paragraph*{Foremost paths in ranges}

Given a temporal path metric $M$, we define the following variant of \textsc{$M$-path TGR} where the input sequence encodes a matrix $D$ of ranges. More precisely, each entry is supposed to represent a \emph{range} $[\ell,r]=\{\ell,\ldots,r\}$ of positive integers. The Ranged-$M$-path predicate is then defined as $P(\G,D):=M(\G)_{uv}\in D_{uv}$ for all $u,v\in [n]$. For example, this leads to the following problem for $M=\Fo$.

\prob{Ranged-Foremost-path TGR}
{A number $n$ and an $n\times n$ matrix $D$ of ranges.}
{Is there a temporal graph $\mg$ such that $\Fo(\G)_{uv}\in D_{uv}$ for all $u,v\in [n]$?}


An entry $(u,v)$ is said to be \emph{undetermined} if $D_{uv}=[\ell,r]$ with $\ell \not= r$. Note that when the number $k$ of undetermined entries is zero, this problem is equivalent to \FoTGR, for which we presented a polynomial-time algorithm.
We now analyze the complexity of~\REARLY and its non-strict variant~\RNSEARLY with a focus on the parameter $k$.

\iflong
\begin{theorem}
\else
\begin{theorem}[$\star$]
\fi
\label{th:ranged-hardness}
\REARLY and \RNSEARLY are both NP-hard. 
Moreover, believing the ETH, neither \REARLY nor \RNSEARLY can be solved in $2^{o(k)}\cdot n^{\Oh(1)}$~time, where~$k$ is the number of undetermined entries of~$D$.
\end{theorem}
\iflong
\begin{proof}
We present a straight forward reduction from~\Reach.

\prob{\Reach}
{A directed graph~$G=(V,A)$.}
{Is there a undirected temporal graph~$\mg$ with \emph{strict reachability graph} equal to~$G$, that is, is there for each pair of distinct vertices~$(u,v)$ a strict temporal path from~$u$ to~$v$ in~$\mg$ if and only if $(u,v)$ is an arc of~$G$?}

Erlebach~et~al.~\cite{EMM25} showed that \Reach is NP-hard and cannot be solved in~$2^{o(|A|)}\cdot n^{\Oh(1)}$~time, unless the ETH fails.
Let~$G=(V,A)$ be an instance of~\Reach.
We define an instance~$D$ of~\REARLY as follows.
For each pair~$(u,v)$ of distinct vertices, we set~$D_{u,v} := [1,n^3]$ if~$(u,v)\in A$ and~$D_{u,v} := n^3+1$, otherwise.
This completes the definition of~$D$.
Note that the number of undetermined entries in~$D$ equals~$|A|$.
We show that~$D$ is realizable if and only if~$G$ is a yes-instance of~\Reach.

$(\Rightarrow)$
Let~$\mg' := (G',\lambda')$ be a realization for~$D$.
Moreover, let~$\mg''$ denote the temporal graph obtained by limiting~$\mg'$ to the first~$n^3$ time steps.
Since~$\mg'$ realizes~$D$, this implies that for each pair~$(u,v)$ of distinct vertices, there is a temporal path from~$u$ to~$v$ in~$\mg''$ if and only if~$D_{u,v} = [1,n^3]$.
The latter is the case if and only if~$(u,v)$ is an arc of~$G$ by definition of $D$.
This implies that the strict reachability graph of~$\mg''$ is equal to~$G$.
Thus, $G$ is a yes-instance of~\Reach.
 
$(\Leftarrow)$
Let~$\mg'=(G'=(V,E'),\lambda')$ be a temporal graph with strict reachability graph that is equal to~$G$.
By Erlebach~et~al.~\cite[Lemma~3]{EMM25}, we can assume that each of the at most~$n^2$ edges of~$G'$ receives at most~$n$ labels each.
Thus, we can assume that the lifetime of~$\mg'$ is at most~$n^3$, since edgeless snapshots can safely be removed. 
Hence, for each pair~$(u,v)$ of distinct vertices, there is a temporal path from~$u$ to~$v$ in~$\mg'$ if and only if~$(u,v)\in A$, which is the case if and only if~$D_{u,v} = [1,n^3]$ by definition of $D$.
This implies that~$\mg'$ realizes all entries of~$D$ with range~$[1,n^3]$ and does not create a temporal path with arrival time less than~$n^3+1$ for any vertex pair~$(u,v)$ with~$D_{u,v} = n^3+1$.
Thus, by additionally assigning label~$n^3+1$ to each edge~$e\in \binom{V}{2}$, the resulting temporal graph also realizes all entries of~$D$ of value~$n^3+1$.
Consequently, $D$ is realizable.

This completes the proof that~\REARLY is NP-hard.
The ETH lower bound now follows from the fact that~\Reach has this ETH lower bound for~$|A|$ and~$|A|$ equals the number~$k$ of undetermined entries in~$D$.

We can also show the hardness and ETH lower bound for the non-strict version, i.e., \RNSEARLY, by instead reducing from~\NReach.
The problem is the same as~\Reach, except that one asks for a temporal graph with~\emph{non-strict} reachability graph equal to~$G$ instead of a temporal graph with strict reachability graph equals to~$G$.
By Erlebach~et~al.~\cite{EMM25} \NReach is also NP-hard and cannot be solved in~$2^{o(|A|)}\cdot n^{\Oh(1)}$~time, unless the ETH fails.
\end{proof}
\fi

For the strict setting, we obtain hardness even when each range has size at most~2.

\iflong
\begin{theorem}
\else
\begin{theorem}[$\star$]
\fi\label{range size 2 hard}
\REARLY is NP-hard even when each range has length at most two and the largest value of the matrix is~$5$. 
\end{theorem}
\iflong
\begin{proof}
We reduce from~\SAT.

Let~$F$ be an instance of~\SAT where no clause contains the same variable both positively and negatively.
Under these restrictions, \SAT is NP-hard~\cite{K72}.

\textbf{Construction.}
Let~$X$ be the variable set of~$F$ and let~$C$ denote the clauses of~$F$.
To obtain an instance~$D$ of~\REARLY, we first define the vertex set over which the matrix~$D$ is defined.
We set~$V := \{x,\overline{x}\mid x\in X\} \cup C \cup \{\top,\bot\}$.

We now describe the entries of~$D$.
See~\Cref{fig ranges} for an illustration.
Let~$c\in C$. 
We set~$D_{c,\top} := 2$ and~$D_{c,\bot} := 3$.
Furthermore, for each variable~$x$ that occurs positively in~$c$, we set~$D_{c,x} := D_{x,c} := 1$ and~$D_{c,\overline{x}} := 2$.
Similarly, for each variable~$x$ that occurs negatively in~$c$, we set~$D_{c,\overline{x}} := D_{\overline{x},c} := 1$ and~$D_{c,x} := 2$.
Let~$x$ be a variable.
We set~$D_{x,\overline{x}} := D_{\overline{x},x} := 1$, $D_{x,\top} := D_{\overline{x},\top} := 2$, and~$D_{x,\bot} := D_{\overline{x},\bot} := 3$.
Furthermore, we set~$D_{\top,x} := D_{\top,\overline{x}} := [2,3]$ and~$D_{\bot,x} := D_{\bot,\overline{x}} := [3,4]$.
All other entries of~$D$ are set to~$5$.
Note that this includes the entry~$D_{\top,\bot}$ and the entry~$D_{\top,c}$ for each clause~$c$.
This completes the definition of~$D$.

\textbf{Intuition.}
The difficulty of the constructed instance comes from the realizability of the entries from clause vertices to~$\top$.
Each clause~$c$ can only reach literals that are contained in~$c$ in time~$1$, but has to reach~$\top$ by time~$2$.
This implies that for some literal~$\ell$ in~$c$, the edge~$\{\ell,\top\}$ needs to receive label~$2$.
Moreover, for each variable~$x$, the entries of~$D$ ensure that at least one of~$\{x,\bot\}$ or~$\{\overline{x},\bot\}$ has to receive label~$3$.
As we do not allow for paths from~$\top$ to~$\bot$ prior to time~$5$, not both edges~$\{x,\top\}$ and~$\{\overline{x},\top\}$ can receive label~$2$, as this would create a path of arrival time~$3$ from~$\top$ to~$\bot$.
Thus, the edges between each variable gadget and the vertex~$\top$ encodes a truth assignment, which we show to be a satisfying assignment of the formula in each realization of~$D$.

\textbf{Correctness.}
We now show that~$F$ is satisfiable if and only if~$D$ is realizable.

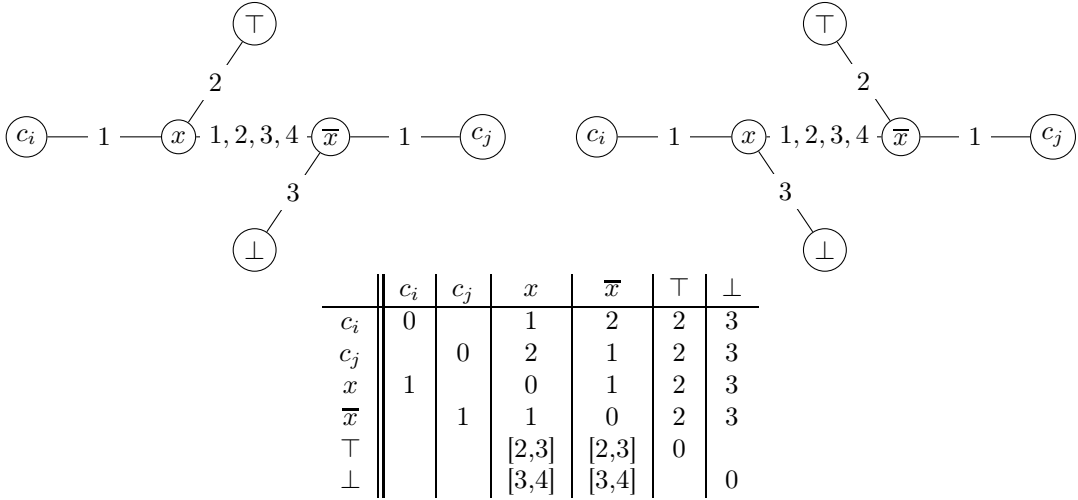
\begin{figure}
\centering
\begin{tikzpicture}

\tikzstyle{k}=[circle,fill=white,draw=black,minimum size=10pt,inner sep=2pt]

\node[k] (l) at (0,0) {$c_i$};
\node[k] (p) at (2,0) {$x$};
\node[k] (n) at (4,0) {$\overline{x}$};
\node[k] (r) at (6,0) {$c_j$};
\node[k] (v) at (3,1.5) {$\top$};
\node[k] (b) at (3,-1.5) {$\bot$};

\draw[-] (l) -- (p) node [midway, fill=white] {$1$};
\draw[-] (r) -- (n) node [midway, fill=white] {$1$};
\draw[-] (n) -- (p) node [midway, fill=white] {$1,2,3,4$};

\draw[-] (p) -- (v) node [midway, fill=white] {$2$};
\draw[-] (n) -- (b) node [midway, fill=white] {$3$};

\begin{scope} [xshift=7.5cm]

\node[k] (l) at (0,0) {$c_i$};
\node[k] (p) at (2,0) {$x$};
\node[k] (n) at (4,0) {$\overline{x}$};
\node[k] (r) at (6,0) {$c_j$};
\node[k] (v) at (3,1.5) {$\top$};
\node[k] (b) at (3,-1.5) {$\bot$};

\draw[-] (l) -- (p) node [midway, fill=white] {$1$};
\draw[-] (r) -- (n) node [midway, fill=white] {$1$};
\draw[-] (n) -- (p) node [midway, fill=white] {$1,2,3,4$};

\draw[-] (n) -- (v) node [midway, fill=white] {$2$};
\draw[-] (p) -- (b) node [midway, fill=white] {$3$};

\end{scope}
\end{tikzpicture}

\begin{tabular}{c||c|c|c|c|c|c}
 & $c_i$ & $c_j$ & $x$ & $\overline{x}$ & $\top$ & $\bot$ \\\hline
$c_i$ & 0 &   & 1 & 2 & 2 & 3  \\
$c_j$ &   & 0 & 2 & 1 & 2 & 3  \\
$x$ & 1 &   & 0 & 1 & 2 & 3  \\
$\overline{x}$ &   & 1 & 1 & 0 & 2 & 3  \\
$\top$ &   &   & [2,3]  & [2,3]  & 0 &  \\
$\bot$ &   &   & [3,4] & [3,4] &  & 0
\end{tabular}
\caption{The variable gadget from the reduction of~\Cref{range size 2 hard} with a true assignment on the left side and a false assignment on the right side. The matrix only shows the entries of value smaller than 5. 
Each clause vertex~$c_i$ aims to reach~$\top$ by time~$2$ and~$\bot$ by time~$3$, but $\top$ and~$\bot$ should pairwise not reach each other prior to time 5.}
\label{fig ranges}
\end{figure}

$(\Leftarrow)$
Let~$\mg:=(G,\lambda)$ be a temporal graph that realizes~$D$, where~$G$ is a static graph with vertex set~$V$.
We show that there is a satisfying assignment for~$F$.
To this end, we analyze the labels on edges incident with~$\top$ and~$\bot$.
Since each entry in~$D$ has value at most~$5$, we can assume without loss of generality that~$G$ is a clique and each edge of~$G$ exists in time step~$5$.
Consider the edges~$E := \{e\in \binom{V}{2}, \min \lambda(e) < 5\}$.
Since~$\mg$ realizes~$D$, for each edge~$\{u,v\}\in E$, we get that~$D_{u,v} < 5$ and~$D_{v,u} < 5$.
Since~$D_{\top,\bot} = 5$ and for each clause~$c$, $D_{\top,c} = D_{\bot,c} = 5$, the only edges of~$E$ incident with~$\top$ or~$\bot$ have literals as  their other endpoint.

\begin{claim}\label{only one two}
For each variable~$x\in X$, $2\notin \lambda(\{x,\top\})$ or~$2\notin \lambda(\{\ol,\top\})$.
\end{claim}
\begin{claimproof}
Let~$x\in X$.
Since~$\mg$ realizes~$D$ and~$D_{x,\bot} = 3$, there is a temporal path~$P$ from~$x$ in~$\mg$ that reaches~$\bot$ at time~$3$.
As already discussed, the only edges incident with~$\bot$ that can receive a label smaller than~$5$ are edges for which the other endpoint is a vertex corresponding to a literal.
Thus, the predecessor of~$\bot$ in~$P$ is a literal~$\ell$.
As~$D_{x,\ell'} = 5$ for each literal~$\ell' \notin \{x,\ol\}$, $\ell$ is either~$x$ or~$\ol$.
This implies that~$3\in \lambda(\{x,\bot\})$ or~$3\in \lambda(\{\ol,\bot\})$.
If~$3\in \lambda(\{x,\bot\})$, then~$2\notin \lambda(\{x,\top\})$, as otherwise, there is a temporal path from~$\top$ in~$\mg$ that reaches~$\bot$ at time~$3 < 5 = D_{\top,\bot}$.
Similarly, if~$3\in \lambda(\{\ol,\bot\})$, then~$2\notin \lambda(\{\ol,\top\})$, as otherwise, there is a temporal path from~$\top$ in~$\mg$ that reaches~$\bot$ at time~$3 < 5 = D_{\top,\bot}$.
Consequently, $2\notin \lambda(\{x,\top\})$ or~$2\notin \lambda(\{\ol,\top\})$.
\end{claimproof}

Based on this property, we now define a truth assignment.
For each variable~$x\in X$, we set variable~$x$ to True if and only if~$2\in \lambda(\{x,\top\})$.
We show that this assignment satisfies~$F$.
To this end, we show that each clause is satisfied by the assignment.
Let~$c$ be a clause.
Since~$D_{c,\top} = 2$ and~$\mg$ realizes~$D$, there is a temporal path~$P$ from~$c$ that reaches~$\top$ at time~$2$.
By $D_{\top,c} = 5$, we get $\min \lambda(\{c,\top\}) = 5$.
This implies that $P$ has length~$2$.
Hence, there is a vertex~$q$, such that~$P=(c,q,\top)$.
Moreover, $1\in \lambda(\{c,q\})$ and~$2\in \lambda(\{q,\top\})$, since we consider strict temporal paths and~$P$ ends in time~$2$.
By~$1\in \lambda(\{c,q\})$, we get that~$D_{c,q} = D_{q,c} = 1$.
This property only holds for the vertices~$q$ that represent the literals that are contained in clause~$c$.
That is, $q = \ell$ for some literal that is contained in~$c$.
Recall that~$P$ reaches~$\top$ at time~$2$.
This implies that~$2 \in \lambda(\{\ell,\top\})$.
Hence, if~$\ell$ is a positive literal~$x$, then variable~$x$ is set to True by the truth assignment.
Otherwise, if~$\ell$ is a negative literal~$\ol$, then the variable~$x$ is set to False by the truth assignment, due to~\Cref{only one two} and the fact that~$2\in \lambda(\{\ell,\top\})$.
In both cases, clause~$c$ is satisfied by the truth assignment, which implies that~$F$ is satisfied.

$(\Rightarrow)$
Let~$\pi$ be a satisfying truth assignment for~$F$.
We define a labeling~$\lambda$ for the complete graph~$G$ with vertex set~$V$, such that~$\mg:= (G,\lambda)$ realizes~$D$.
We initialize~$\lambda$ by setting~$\lambda(e) := \{5\}$ for each edge~$e$ of the complete graph~$G$.
Note that this guarantees that between each pair of vertices there is always a temporal paths with arrival time at most~$5$. 
Since this is the largest entry in~$D$, we guarantee to not realize journeys of too early arrival time in this way.
In the following, we describe how to add further labels to the edges, to realize all entries smaller than~$5$ in~$D$.
Let~$x$ be a variable.
We add the labels~$1,2,3$, and~$4$ to the edge~$\{x,\ol\}$.
If~$x$ is assigned to True by~$\pi$, we add label~$2$ to edge~$\{x,\top\}$ and label~$3$ to edge~$\{\ol,\bot\}$.
Otherwise, that is, if~$x$ is assigned to False by~$\pi$, we add label~$2$ to edge~$\{\ol,\top\}$ and label~$3$ to edge~$\{x,\bot\}$.
Finally, if variable~$x$ occurs positively in a clause~$c$, we add label~$1$ to the edge~$\{x,c\}$, and if variable~$x$ occurs negatively in a clause~$c$, we add label~$1$ to the edge~$\{\ol,c\}$.
This completes the definition of~$\lambda$.
We now show that this labeling realizes~$D$.

Note that (i)~only the labels~$1$ and~$5$ are on edges incident with clause vertices, (ii)~only the labels~$2$ and~$5$ are on edges incident with vertex~$\top$, and (iii)~only the labels~$3$ and~$5$ are on edges incident with vertex~$\bot$.
This implies that each temporal path~$P$ with arrival time smaller than~$5$ only uses vertices that correspond to literals as internal vertices.
Moreover, since for any two distinct variables~$x$ and~$y$, there is no edge with label smaller than~$5$ between any vertex of~$\{x,\ol\}$ and any vertex of~$\{y,\ol[y]\}$, each temporal path with arrival time less than~$5$ uses only edges of the subgraph depicted in~\Cref{fig ranges} for some variable~$x$.
That is, each such temporal path~$P$ uses only edges of~$G_x$ for some variable~$x$, where~$G_x := G[\{x,\ol,\top,\bot\} \cup C_x \cup C_{\ol}]$ with~$C_x := \{c\in C \mid x\in c\}$ and~$C_{\ol}:= \{c\in C \mid \ol\in c\}$.
Note that when restricting the edges of~$G_x$ to those that receive at least one label smaller than~$5$, the resulting static graph is a tree.
It is easy to verify that no temporal path in this tree arrives faster than specified in~$D$.
In the following, we thus only focus on showing that there is always a temporal path in~$\mg$ with exactly the arrival time as specified by~$D$.
We only focus on those entries of value smaller than~$5$, as the entries of value~$5$ are clearly realized by the fact that each edge receives label~$5$.

Let~$x$ be a variable.
Since~$1\in \lambda(\{x,c\})$ for each clause~$c\in C_x$, the entries~$D_{c,x}$ and~$D_{x,c}$ of value~$1$ are realized.
Similarly, since~$1\in \lambda(\{\ol,c\})$ for each clause~$c\in C_{\ol}$, the entries~$D_{c,\ol}$ and~$D_{\ol,c}$ of value~$1$ are realized.
Moreover, since the edge~$\{x,\ol\}$ receives labels~$1$ and~$2$, the entry~$D_{x,\ol}$ of value~$1$ are realized, as well as the entries~$D_{c,\ol}$ and~$D_{c',x}$ of value~$2$ for each clause~$c\in C_x$ and each clause~$c'\in C_{\ol}$.
Note that this implies that all entries of value less than~$5$ involving only clause and literal vertices are realized.
Now, recall that~$3\in \lambda(\{x,\bot\})$ or~$3\in \lambda(\{\ol,\bot\})$.
Since~$x$ and~$\ol$ can pairwise reach each other in time~$1$, and each vertex of~$C_x \cup C_{\ol}$ can reach vertex~$x$ and vertex~$\ol$ at time at most~$2$, there are paths that arrive at vertex~$\bot$ at time~$3$ for all vertices of~$\{x,\ol\} \cup C_x \cup C_{\ol}$. 
This realizes all entries~$D_{q,\bot}$ of value~$3$ with~$q\in \{x,\ol\} \cup C_x \cup C_{\ol}$.
Since each clause contains at least one variable, this thus implies that all entries of the form~$D_{v,\bot}$ are realized for each vertex~$v$ of~$G$.
Now consider the entries~$D_{\bot,x}$ and~$D_{\bot,\ol}$.
By~$3\in \lambda(\{x,\bot\})$ or~$3\in \lambda(\{\ol,\bot\})$, vertex~$\bot$ can reach one of the vertices~$x$ or~$\ol$ at time~$3$ and the other vertex at time~$4$, since edge~$\{x,\ol\}$ has label~$4$.
Thus, the entries~$D_{\bot,x}$ and~$D_{\bot,\ol}$ of range~$[3,4]$ are realized.
It remains to consider the entries involving~$\top$.
Recall that~$2\in \lambda(\{x,\top\})$ \jc{or}~$2\in \lambda(\{\ol,\top\})$.
By the fact that edge~$\{x,\ol\}$ receives label~$1$ and~$3$, this implies that (i)~both~$x$ and~$\ol$ can reach~$\top$ at time~$2$ and (ii)~vertex~$\top$ can reach one of the vertices~$x$ or~$\ol$ at time~$2$ and the other one at time~$3$.
Thus, (i)~the entries~$D_{x,\top}$ and~$D_{\ol,\top}$ of value~$2$ are realized and (ii)~the entries~$D_{\top,x}$ and~$D_{\top,\ol}$ of range~$[2,3]$ are realized.
Hence, all entries involving~$x$ or~$\ol$ are realized.
It remains to consider the entry~$D_{c,\top}$ of value~$2$ for each clause~$c\in C$.
Let~$c\in C$.
Since the truth assignment~$\pi$ satisfies~$F$, there is a literal~$\ell$ contained in~$c$ that is assigned to True by~$\pi$.
By definition of~$\lambda$, $1\in \lambda(\{c,\ell\})$ and~$2\in \lambda(\{\ell,\top\})$.
This implies that vertex~$c$ can reach~$\top$ at time~$2$ and thus that the entry~$D_{c,\top}$ of value~$2$ is realized.
Hence, all entries of value smaller than~$5$ of~$D$ are realized by~$\mg$.
Consequently, $\mg$ realizes~$D$.
\end{proof}
\fi

\paragraph*{FPT algorithm for \textsc{Ranged-Foremost-path TGR}}




We now propose a dynamic programming algorithm solving \textsc{Ranged-Foremost-path TGR} which has running time~$2^{\Oh(k)}\cdot n^{\Oh(1)}$. 
Note that this is tight in the sense that a significantly faster algorithm would contradict ETH by \Cref{th:ranged-hardness}.

\iflong
\begin{theorem}
\else
\begin{theorem}[$\star$]
\fi
\label{th:ranged-foremost}
	\textsc{Ranged-Foremost-path TGR} can be solved in $\Oh(k^2 3^kn^4)$ time and $\Oh(2^k+n^2)$ space, where $k$ is the number of undetermined entries of $D$. 
\end{theorem}

In the following, we let $\ell_{uw}$ (resp. $r_{uw}$) denote the lower bound (resp. upper bound) of entry $(u,w)$, i.e., $D_{uw}=[\ell_{uw},r_{uw}]$. Recall that an entry $(u,w)$ is undetermined if $\ell_{uw}\neq r_{uw}$. The set of such undetermined entries is denoted by $Undet$ and its size is denoted by $k$.

\iflong
The above result is a consequence of Algorithm~\ref{alg:ranged-foremost}, which is based on the following observation.

\begin{lemma}\label{lem:ranged-time}
	If there exists a realization $\G$ of $D$ for Ranged-$\Fo$-path, then there is a realization $\G'$ of $D$ for Ranged-$\Fo$-path such that all time labels of $\G'$ are in the set $\{\ell_{uw}+j\mid u\neq w \text{ and } 0\leq j \leq k\}$, where $k$ is the number of undetermined entries of~$D$.
\end{lemma}

\begin{proof}
Let $\G$ be a realization of $D$ for Ranged-$\Fo$-path. Without loss of generality, we can assume that all time labels of $\G$ are coefficients of its foremost matrix $\Fo(\G)$. Indeed, let $(\{u,v\},\tau)$ be a temporal edge  of $\G$ whose time label $\tau$ is not a coefficient of its foremost matrix $\Fo(\G)$. Let $P$ be a foremost temporal path from a vertex $x$ to a vertex $y$ that traverses $\{u,v\}$ from $u$ to $v$ at time $\tau$, without loss of generality. We have that there exists a foremost temporal $xv$-path $P'$ arriving at time $\Fo(\G)_{xv}<\tau$ as $\tau$ is not a coefficient of $\Fo(\G)$. The temporal path $P'$ extended with the suffix of $P$ that starts at node $v$ is thus a foremost temporal $xy$-walk that does not use temporal edge $(\{u,v\},\tau)$. 
Thus, deleting temporal edge $(\{u,v\},\tau)$ preserves the fact that $\G$ is a realization of $D$ for Ranged-$\Fo$-path. 

Thus, let $\mathcal{G}=(G,\lambda)$ be a realization of $D$ such that all the time labels of $\mathcal{G}$ are coefficients of its foremost matrix $\Fo(G)$. Let $\ell_1<\dots<\ell_p$ be the lower bounds of the entries of $D$ sorted in increasing order (excluding 0 and $\infty$). Let $C_\G$ be the set of coefficients of the foremost matrix of $\G$ that are not 0 or $\infty$. For all $i\in [p]$, we denote $\ell_i^1,\dots,\ell_i^{k_i}$ the coefficients in $C_\G$ sorted in increasing order that are strictly between $\ell_i$ and $\ell_{i+1}$, with the convention that $\ell_{p+1}=\infty$. 
In other words, we have the following increasing sequence: $0<\ell_1<\ell_1^1<\dots<\ell_1^{k_1}<\ell_2<\dots< \ell_i<\ell_i^1<\dots<\ell_i^{k_i}<\ell_{i+1}<\dots<\ell_{p+1}=\infty$. 
First, observe that, for all $i\in [p]$, the number $k_i$ of coefficients in $C_\G$ that are between the lower bounds $\ell_i$ and $\ell_{i+1}$ is at most $k$, where $k$ is the number of undetermined entries of $D$. 
We construct the temporal graph $\G'$ as follows: for each $i\in [p]$ and $j\in [k_i]$, we replace all time labels $\ell_i^j$ of $\G$ by the label $\ell_i+j$. 
We then have that all time labels of $\G'$ are in the set $\{\ell_{uw}+j \mid u\neq w \text{ and } 0\leq j \leq k\}$. Let us prove that $\G'$ is still a realization of $D$. We show that for all $(u,w)\in [n]\times [n]$, if $\Fo(\G)_{uw}=\ell_i^j$ for some $i\in [p]$ and $j\in [k_i]$ then $\Fo(\G')_{uw}=\ell_i+j$, and $\Fo(\G')_{uw}=\Fo(\G)_{uw}$ otherwise. 
This comes from the fact that the order of the sequence $0<\ell_1<\ell_1^1<\dots<\ell_1^{k_1}<\ell_2<\dots< \ell_i<\ell_i^1<\dots<\ell_i^{k_i}<\ell_{i+1}<\dots<\ell_{p+1}=\infty$ is preserved, i.e. $0<\ell_1<\ell_1+1<\dots<\ell_1+{k_1}<\ell_2<\dots< \ell_i<\ell_i+1<\dots<\ell_i+{k_i}<\ell_{i+1}<\dots<\ell_{p+1}=\infty$. The transformation thus preserves temporal paths. It remains to prove that, for each $(u,w)$ such that $\Fo(\G)_{uw}=\ell_i^j$ and $\Fo(\G')_{uw}=\ell_i+j$ for some $i\in [p]$ and $j\in [k_i]$, $\ell_i+j$ is in $D_{uw}$. As $\Fo(\G)_{uw}=\ell_i^j$ and $\G$ realizes $D$, we have $\ell_{uv}\leq \ell_i^j\leq r_{uw}$. 
As $\ell_{uv}$ is a lower bound of an entry of the matrix $D$, we also have that $\ell_{uv}\leq \ell_i$. We conclude with $\ell_{uv}\leq \ell_i\leq \ell_i+j\leq \ell_i^j\leq r_{uw}$.  
\end{proof}
\else
\fi


\iflong
The general idea of the algorithm is to process all possible times in 
$\{\ell_{uw}+j\mid u\neq w \text{ and } 0\leq j \leq k\}$ 
\else
This result relies on the fact that a realizable instance can always be realized by a temporal graph using time labels in the restricted set
$\T=\{\ell_{uw}+j\mid u\neq w \text{ and } 0\leq j \leq k\}$ (as proven in appendix). We then propose an algorithm that processes all times in $\T$
\fi
in increasing order and guesses which undetermined entries can be realized at the current time. More precisely, letting $0<\tau_1<\cdots<\tau_m$ denote these times, we maintain, for each $i\in [m]$, a table $R[\cdot,i]$ such that $R[S,i]$ for $S\subseteq Undet$ is equal to True if and only if there exists a temporal graph $\G$ with time labels in $\{\tau_1,\dots,\tau_i\}$ such that:
\begin{itemize}
\item for each $(u,w)\in S$, the earliest arrival time of any foremost temporal $uw$-path in $\G$ is in $D_{uw}$ and is at most $\tau_i$,
\item for each $(u,w)$ with $\ell_{uw}=r_{uw}=\tau\leq\tau_i$, the earliest arrival time of any foremost temporal $uw$-path in $\G$ is $\tau$,
\item for all other entries $(u,w)$, there is no temporal $uw$-path in $\G$.
\end{itemize} 
\iflong
As a consequence of Lemma~\ref{lem:ranged-time}, we
\else
We
\fi
have that $R[Undet,m]=True$ if and only if there exists a temporal graph that realizes $D$ for Ranged-$\Fo$-path.
 
When processing time $\tau_i$, we consider tri-partitions $S,T,U$ of the set $Undet$  where $S$ represents the set of undetermined entries that must be realized before time $\tau_i$, $T$ represents the set of undetermined entries that must be realized at time $\tau_i$ and $U$ the set of undetermined entries that must be realized after time $\tau_i$. 
\iflong
We first check that $R[S,i-1]$ is equal to True and that all entries $(u,w)$ in $T$ are such that $\tau_i\in D_{uw}$. If this is the case, we have to check that we can realize the entries in $T$ with time $\tau_i$, as well as entries $(u,w)$ such that $l_{uw}=r_{uw}=\tau_i$, by completing a temporal graph associated to $R[S,i-1]$. To do so, we proceed similarly to \Cref{alg:foremost}, by looking for a suitable vertex $v$ when dealing with the entry $(u,w)$, such that adding the temporal edge $(\{v,w\},\tau_i)$ creates a temporal $uw$-path with arrival time $\tau_i$ without creating any temporal path that arrives too early. This is verified in the algorithm using the following definition of edge compatibility:
\begin{itemize}
	\item $RangeEdgeCompat(D,\{v,w\},\tau,S,T):=RangeEdgeCompatDir(D,v,w,\tau,S,T)$ and $RangeEdgeCompatDir(D,w,v,\tau,S,T)$,
	\item where $RangeEdgeCompatDir(D,v,w,\tau,S,T):=\forall x \in [n], (r_{xv}<\tau \text{ or } (x,v)\in S) \Longrightarrow (r_{xw}\leq\tau \text{ or } (x,w)\in S\cup T)$.
\end{itemize}
If we can indeed complete such a temporal graph, then we set $R[S\cup T,i]$ to True. Note that, similarly to \Cref{alg:foremost}, we do not need explicit access to such a temporal graph as we probe it through $D,S,T$.
\else
Similarly to \Cref{alg:foremost}, an appropriate definition of edge compatibility (with respect to $D$, $S$ and $T$) allows to test if a temporal graph realizing $R[S,i-1]=True$ can be completed in order to set $R[S\cup T,i]$. Importantly, we do not need explicit access to such a temporal graph as we probe it through $D,S,T$.
\fi

\iflong
The correctness of the algorithm relies on the following lemma about edge compatibility:

\begin{lemma}\label{lem:ranged-compat}
	If $\G$ is a realization of $D$ for Ranged-$\Fo$-path, then for all temporal edges $(\{v,w\},\tau)$ in $\G$, we have that $RangeEdgeCompat(D,\{v,w\},\tau,S,T)$ is satisfied, where $S=\{(u,w)\in Undet\mid  \Fo(\G)_{uw} <\tau\}$ is the set of undetermined entries whose earliest arrival time in $\G$ is before $\tau$ and $T=\{(u,w)\in Undet\mid  \Fo(\G)_{uw}=\tau\}$ is the set of  undetermined entries whose earliest arrival time in $\G$ is at $\tau$.
\end{lemma}

\begin{proof}
Suppose for the sake of contradiction that some temporal edge $(\{v,w\},\tau)$ does not satisfy $RangeEdgeCompat(D,\{v,w\},\tau,S,T)$. That is, without loss of generality, there exists a vertex $x$ such that $(r_{xv}<\tau \text{ or } (x,v)\in S)$ and $(r_{xw}>\tau \text{ and } (x,w)\notin S\cup T)$. By definition of the set $S$, and since $\G$ is a realization of $D$, the condition $(r_{xv}<\tau \text{ or } (x,v)\in S)$ means that there exists a foremost temporal $xv$-path $P$ arriving at a time $\sigma<\tau$. 
Extending this path with $(\{v,w\},\tau)$ yields a temporal $xw$-walk arriving at time $\tau$. If $(x,w)$ is an undetermined entry of $D$, this contradicts the fact that $(x,w)\notin S\cup T$. Otherwise, we have $\ell_{xw}=r_{xw}$, and $\ell_{xw}>\tau$ contradicts the fact that $\G$ is a realization of $D$ for Ranged-$\Fo$-path.
\end{proof}

\begin{algorithm}[ht]
	\caption{Ranged-Foremost-path TGR}\label{alg:ranged-foremost}
	\begin{algorithmic}[1]
	\Require A number $n$, an $n\times n$ matrix $D$ of ranges
	\Ensure YES if there exists a temporal graph that realizes $D$ for Ranged-$\Fo$-path, NO otherwise
	
	\State Let $Undet=\{(u,w)\in [n]\times [n], \ell_{uw}\neq r_{uw}\}$ be the set of undetermined entries of $D$
	\State Let $k=|Undet|$ be the number of undetermined entries of $D$
	\State Let $0<\tau_1<\dots<\tau_m$ be the times in $\{\ell_{uw}+j\mid u\neq w \text{ and } 0\leq j \leq k \text{ and } \ell_{uw}<\infty\}$
	\State Set $R[S,0]:=False$ for all $S\subseteq Undet$
	\State $R[\emptyset,0]:=True$
	\ForAll{$i=1$ to $m$}
		\State $R[S,i]:=False$ for all $S\subseteq Undet$
		\State Let $C_i$ be the set of all $(u,w)$ such that $\ell_{uw}=r_{uw}=\tau_i$
		\ForAll{tri-partition $S,T,U$ of $Undet$ such that $R[S,i-1]$ and $\forall (u,w)\in T, \tau_i\in D_{uw}$} \label{lin:tri-part}
			\If{\textsc{Exists}$(D,C_i,\tau_i,S,T)$}
				\State $R[S\cup T, i]:=True$
			\EndIf
		\EndFor
	\EndFor
	\State Return $R[Undet,m]$

	\smallskip
	
	\Procedure{Exists}{$D,C,\tau,S,T$}
		\Comment{Given an $n\times n$ matrix $D$ of ranges, a time $\tau$, $C,S,T\subseteq [n]\times [n]$}
		
		\ForAll{$(u,w)\in C\cup T$}\label{lin:couple}
			\If{$\nexists v\in [n]$ such that ($r_{uv}< \tau$ or $(u,v)\in S$) and $RangeEdgeCompat(D,\{v,w\},\tau,S,T)$}\label{lin:find}
				\State Return False
			\EndIf
		\EndFor 
		\State Return True
		
	\EndProcedure
	
	\end{algorithmic}
	\end{algorithm}

\begin{proof}[Proof of Theorem~\ref{alg:ranged-foremost}]
Let us first prove that if there exists a temporal graph that realizes $D$ for Ranged-$\Fo$-path, then the algorithm returns YES. Let $\G$ be such a realization with all time labels being in the set $\{\ell_{uw}+j\mid u\neq w \text{ and } 0\leq j \leq k\}$. Such a temporal graph exists by Lemma~\ref{lem:ranged-time}. Let us denote $\tau_1<\dots<\tau_m$ the times in $\{\ell_{uw}+j\mid u\neq w \text{ and } 0\leq j \leq k \text{ and } \ell_{uw}<\infty\}$. Let $S_i=\{(u,w)\in Undet \mid  \Fo(G)_{uw}=\tau_i\}$ be the set of undetermined entries $(u,w)$ such that a foremost temporal $uw$-path in $\G$ arrives at time $\tau_i$. First, observe that $\cup_{k=1}^m S_k = Undet$. Let us prove by induction on $i$ that after the $i$th iteration of the main loop, we have $R[\cup_{k=1}^i S_k,i]=True$. We indeed have $R[\emptyset,0]=True$. For $i\geq 1$, when the tri-partition $\cup_{k=1}^{i-1} S_k,S_i,Undet\setminus \cup_{k=1}^i S_k$ is considered at Line~\ref{lin:tri-part}, we have:
\begin{itemize}
\item $R[\cup_{k=1}^{i-1} S_k,i-1]=True$ by induction hypothesis.
\item $\forall (u,w)\in S_i$, $\tau_i\in D_{uw}$ by definition of $S_i$ and the fact that $\G$ is a Ranged-$\Fo$-realization of $D$.
\item $\textsc{Exists}(D,C_i,\tau_i,\cup_{k=1}^{i-1} S_k,S_i)=True$ where $C_i$ is the set of all $(u,w)$ such that $\ell_{uw}=r_{uw}=\tau_i$. Indeed, let $(u,w)\in C_i \cup S_i$. We have that there exists a foremost temporal $uw$-path $P$ in $\G$ that arrives at time $\tau_i$. Consider the predecessor $v$ of $w$ in $P$. By Lemma~\ref{lem:ranged-compat}, we have that $RangeEdgeCompat(D,\{v,w\},\tau_i,\cup_{k=1}^{i-1} S_k,S_i)$ is satisfied. Moreover, either $(u,v)$ is an undetermined entry, and in this case we have $(u,v)\in \cup_{k=1}^{i-1} S_k$, or $l_{uv}=r_{uv}$ and in this case we have $r_{uv}<\tau$. Thus, the procedure $\textsc{Exists}$ returns True. 
\end{itemize}
As a consequence, $R[\cup_{k=1}^i S_k,i]$ is set to True. This implies that $R[Undet,i]$ is set to True by considering the last index $i=m$. This concludes the proof that the algorithm returns NO only when no realization exists. 

Let us now prove that if the algorithm returns YES, then there exists a temporal graph $\G$ that realizes $D$ for Ranged-$\Fo$-path. If the algorithm returns YES, we have a sequence of subsets of $Undet$ $S_1,\dots,S_m$\lvtodo{would it be more appropriate to use $T_1,\dots,T_m$? JC: Yes, indeed JC: Yes, indeed} such that: 
\begin{itemize}
\item[(i)] $\cup_{k=1}^{m} S_k=Undet$,\lvtodo{use another index than $k$ the parameter}
\item[(ii)] $R[\cup_{k=1}^{i} S_k,i]=True$ for each $i\in[m]$,
\item[(iii)] $\forall (u,w)\in S_i, \tau_i\in D_{uw}$,
\item[(iv)] $\textsc{Exists}(D,C_i,\tau_i,\cup_{k=1}^{i-1} S_k,S_i)=True$ for each $i\in [m]$.
\end{itemize}

Consider the distance matrix $D'$ defined by $D'_{uw}=\tau_i$ for all $i\in [m]$ and $(u,w)\in C_i\cup S_i$. This completely defines $D'$ as $\cup_{i\in [m]}S_i=Undet$ by (i) and $\cup_{i\in [m]}C_i$ is the set of determined entries in $D$. We will show that there exists a $\Fo$-realization of $D'$. As we have $D'_{uw}\in D_{uw}$ for all $u\not= v$ by (iii), this will imply that such a realization is also a realization of $D$ for Ranged-$\Fo$-path.
Now, we note that for $i\in [m]$ and $(u,w)\in C_i\cup S_i$, the test $RangeEdgeCompat(D,\{v,w\},\tau_i,\cup_{k=1}^{i-1} S_k,S_i)$ is equivalent to $EdgeCompat(D',\{v,w\},\tau_i)$ since ($r_{xv}< \tau_i$ or $(x,v)\in \cup_{k=1}^{i-1} S_k$) is equivalent to $D'_{xv} < \tau_i$ by definition of $D'$. 
Similarly, ($r_{uv}< \tau_i$ or $(u,v)\in \cup_{k=1}^{i-1} S_k$) is equivalent to $D'_{uv} < \tau_i$.
The tests performed at Line~\ref{lin:find} by the $m$ calls to $\textsc{Exists}(D,C_i,\tau_i,\cup_{k=1}^{i-1} S_k,S_i)$ are thus equivalent to all tests at Line~\ref{lin:test-inf} of \Cref{alg:foremost} running with input $D'$. By its correctness (see the proof of \Cref{th:foremost}), (iv) implies that $D'$ is $\Fo$-realizable.

Regarding the time complexity, for $i\in [m]$, \textsc{Exists}$(D,C_i,\tau_i,S,T)$ for a tri-partition $S,T,U$ of $Undet$ takes $\mathcal{O}((|C_i|+k)n^2)$ time assuming that we can check in constant time whether an undetermined entry $(u,v)$ is in set $S$ or $T$. This is achieved by using $n\times n$ boolean matrices for storing $S$ and $T$ in $O(n^2)$ time for each tri-partition (and $O(n^2)$ space). The overall time complexity is thus $\Oh(\sum_{i=1}^m 3^k(|C_i|+k)n^2) = \Oh(3^kn^4+3^k k n^2m) = \Oh(k^2 3^kn^4)$ as $m = \Oh(kn^2)$.
Regarding the space complexity, note that at step $i\in [m]$, only the tables $R[\cdot,i-1]$ and $R[\cdot,i]$ are used. We can thus maintain at each step only 2 such tables. 
The overall space complexity is thus in $\Oh(2^k+n^2)$. 
\end{proof}
\fi
\iflong

Note that the definition of $RangeEdgeCompat$ can be adapted in order to deal with the non-strict variant with ranges and the prescribed variant with ranges. 
Namely, we can respectively define:
\begin{itemize}
\item $RangeNSEdgeCompat(D,\{v,w\},\tau,S,T):=\forall x \in [n], (r_{xv}\leq\tau \text{ or } (x,v)\in S\cup T) \Longleftrightarrow (r_{xw}\leq\tau \text{ or } (x,w)\in S\cup T)$,
\item $RangePrEdgeCompatDir(D,v,w,\tau,S,T,G_p=([n],E_p)):=\forall x \in [n], (r_{xv}<\tau \text{ or } (x,v)\in S) \Longrightarrow (r_{xw}\leq\tau \text{ or } (x,w)\in S\cup T)$ and $\{v,w\}\in E_p$.
\end{itemize}
\lvnew{Note also that it is straightforward to adapt the proof of \Cref{lem:ranged-time} to obtain a similar statement for Ranged-$\NSFo$-path.}
\fi

\subparagraph{Remark.} The algorithm proposed here can easily be generalized to a more general setting where a collection of ranges is given for each entry of $D$  with time complexity $\Oh(k^2 3^kn^2N)$ where $N\ge n^2$ denotes the total number of ranges in $D$. In particular, this provides an FPT algorithm for \REARLY when each entry of $D$ encodes a set of integers and at most $k$ of them are non-singletons.
Note that the hardness result of \Cref{range size 2 hard} holds in that setting, even if each set has size at most 2.

\section{Fastest paths}\label{sec fast}
In this section we revisit temporal graph realization for fastest paths analyzed by Klobas et al.~\cite{KMMS24} and Erlebach et al.~\cite{EMW24}.
We answer an open question by both papers about the parameterized complexity with respect to the vertex cover number.
\iflong
Moreover, we consider also non-strict paths and the non-periodic setting with arbitrary many labels per edge.

Recall the problem definition:
\prob{\FAST}
{A duration matrix $D$ of size $n\times n$.}
{Is there a temporal graph $\mg$ (with unbounded lifetime and unbounded number of labels) that realizes the duration matrix $D$, that is, for each ordered pair of vertices $(u,v)$ in $\mg$ with $u\not= v$, the duration of a fastest past from~$u$ to~$v$ equals~$D_{uv}$?}

\else
\fi
So far, \FAST has only been considered for strict temporal paths and if we consider a periodic temporal graph or if we limit the number of labels per edge (per period)~\cite{KMMS24,EMW24}.
Our parameterized hardness result holds even for non-periodic temporal graph with arbitrary lifetime and arbitrary many labels per edge.
We then lift this hardness also to the periodic case with one label per edge per period to answer the open questions.

\iflong
\paragraph*{Strict fastest paths}
\fi
\iflong
\begin{theorem}
\else
\begin{theorem}[$\star$]
\fi\label{hardness fast strict}
\FAST 
is NP-hard and W[1]-hard when parameterized by the vertex cover number of the underlying graph plus the largest entry of~$D$.
This holds even on a family of instances for which all yes-instances are realizable with only one label per edge.
\end{theorem}

\iflong
\begin{proof}
\else
\begin{proof}[Proof (sketch).]
\fi
We reduce from \MCC~\cite{C+15}.

\prob{\MCC}{An undirected graph~$G=(V,E)$, an integer~$k$, and a~$k$-partition~$(V_1 \cup \dots \cup V_k)$ of~$V$, such that~$V_i$ is an independent set in~$G$ for each~$i\in [1,k]$.}{Is there a clique of size~$k$ in~$G$?}

\iflong
\lvnew{Note that a clique in such a graph is said to be multicolored since its vertices must belong to pairwsie-distinct sets among $V_1 \cup \dots \cup V_k$. The reason is that these sets are pairwise-disjoint independent sets, and they can be viewed as color classes.}
\fi
Let~$I:=(G=(V_1\cup\dots\cup V_k, E), k)$ be an instance of~\MCC where for each~$1\leq a < b \leq k$, $G[V_a\cup V_b]$ is a disjoint union of bicliques.
For each~$a\in [1,k]$, we call~$V_a$ a~\emph{color class}.
Even under these restrictions, \MCC is NP-hard and W[1]-hard when parameterized by~$k$~\cite{MRW20}.
Let~$V:= V_1 \cup \dots \cup V_k$.

\iflong
\subparagraph{Construction.}
\fi
To obtain an instance~$D$ of \FAST as follows, we first describe the underlying graph, that is, the graph~$G'$ that contains an edge~$\{u,v\}$ if and only if~$D_{u,v} = 1$.
The graph~$G'$ is defined over the vertex set~$V' := V \cup \{s,s',s'',t,t',t''\} \cup X \cup L$, where~$L$ is a vertex set of size~$2k+2$ and~$X:= \{x_i\mid 1 \leq i \leq k+1\}$ (see~\Cref{fig fastest graph}). 
We add edges between these vertices, such that~$V' \setminus V$ is a vertex cover of~$G'$.
That is, there are no edges between the vertices of~$V$ in~$G'$.
We make~$L$ into a clique and adjacent to all vertices of~$V'$ besides~$s$ and~$t$.
Similarly, we make the vertices~$s'$ and~$s''$ adjacent to all vertices of~$V'$ besides~$t$ and make the vertices~$t'$ and~$t''$ adjacent to all vertices of~$V'$ besides~$s$.
Additionally, we make~$x_1$ adjacent to the vertices of~$V_1$, $x_{k+1}$ adjacent to the vertices of~$V_k$, and for each~$i\in [2,k]$, we make~$x_{i}$ adjacent to the vertices of~$V_{i-1}\cup V_i$.
There are no edges between the vertices of~$X$.
Finally, we make~$s$ adjacent to~$s',s''$, and~$x_1$, and we make~$t$ adjacent to~$t',t''$, and~$x_{k+1}$.

This completes the underlying graph~$G'$ and thus all entries of the matrix~$D$ of value~$1$.
Let~$E'$ denote the edges of~$G'$.
Next, we define the remaining entries.
Note that each vertex of~$\{s',s'',t',t''\} \cup L$ has an edge towards each other vertex of~$V'\setminus \{s,t\}$.
Hence, for these vertices it remains to define the entries of the table from and to the vertices~$s$ and~$t$.
We set~$D_{s,q} := D_{q,s} := 2k+3$ for each vertex~$q\in \{t',t''\} \cup L$.
Similarly, we set~$D_{t,q} := D_{q,t} := 2k+3$ for each vertex~$q\in \{s',s''\} \cup L$.
For each non-edge~$\{u,v\}\notin E$, we set~$D_{u,v} := D_{v,u} := 2k+3$.
Finally, we set~$D_{s,t}:=2k+2$ and~$D_{t,s}:=4k+5$.
All other undefined entries are set to~$2$.
This completes the construction.
Note that $X \cup L \cup \{s',s'',t',t''\}$ is a vertex cover of size $3k+7$ and that the largest entry in $D$ is $4k+5$.

\begin{figure}
\centering
\begin{tikzpicture}[scale=.65]

\tikzstyle{k}=[circle,fill=white,draw=black,minimum size=8pt,inner sep=2pt]

\node[k] (s) at (0,2) {$s$};

\begin{scope}[xshift=1.5cm]
\node[k] (x1) at (0,0) {$x_1$};
\end{scope}

\begin{scope}[xshift=3cm]
\def\y{1}

\draw[rectangle] (.3,-2) -- (-.3,-2) -- (-.3,2) -- (.3,2) -- (.3,-2);
\node (v1) at (0,-2.4) {$V_{\y}$};

\foreach \x in {1,2,3,4,5} 
\node[k] (\y\x) at ($(0,-2.25) +  (0,\x*.75)$) {};
\end{scope}

\begin{scope}[xshift=4.5cm]
\node[k] (x2) at (0,0) {$x_2$};
\end{scope}

\begin{scope}[xshift=6cm]
\def\y{2}

\draw[rectangle] (.3,-2) -- (-.3,-2) -- (-.3,2) -- (.3,2) -- (.3,-2);
\node (v1) at (0,-2.4) {$V_{\y}$};

\foreach \x in {1,2,3,4,5} 
\node[k] (\y\x) at ($(0,-2.25) +  (0,\x*.75)$) {};
\end{scope}

\begin{scope}[xshift=7.5cm]
\node (x3) at (0,0) {};
\end{scope}

\begin{scope}[xshift=8.25cm]
\node (xxx) at (0,0) {\huge $\cdots$};
\end{scope}

\begin{scope}[xshift=9cm]
\node (xk) at (0,0) {};
\end{scope}

\begin{scope}[xshift=10.5cm]

\def\y{k}

\draw[rectangle] (.3,-2) -- (-.3,-2) -- (-.3,2) -- (.3,2) -- (.3,-2);
\node (v1) at (0,-2.4) {$V_{\y}$};

\foreach \x in {1,2,3,4,5} 
\node[k] (\y\x) at ($(0,-2.25) +  (0,\x*.75)$) {};

\begin{scope}[xshift=1.5cm]
\node[k] (xk1) at (0,0) {$x_{k+1}$};
\end{scope}

\begin{scope}[xshift=3cm]
\node[k] (t) at (0,2) {$t$};
\end{scope}
\end{scope}

\node[k] (s1) at ($(x1) + (0,3)$) {$s'$};
\node[k] (s2) at ($(x1) + (0,4)$) {$s''$};
\node[k] (t1) at ($(xk1) + (0,3)$) {$t'$};
\node[k] (t2) at ($(xk1) + (0,4)$) {$t''$};

\node[k,inner sep=7pt] (L) at (6.75,5) {$L$};

\draw[-] (s) to (s1);
\draw[-] (s) to (s2);
\draw[-] (t) to (t1);
\draw[-] (t) to (t2);

\draw[-] (t1) to (t2);
\draw[-] (s1) to (s2);
\draw[-] (s1) to (t1);
\draw[-] (s1) to (t2);
\draw[-] (s2) to (t1);
\draw[-] (s2) to (t2);

\draw[-] (s1) to (L);
\draw[-] (s2) to (L);
\draw[-] (t1) to (L);
\draw[-] (t2) to (L);

\draw[-] (s) to (x1);
\draw[-] (xk1) to (t);

\foreach \x in {1,2,k}
\foreach \y in {1,2,3,4,5}
\draw[-] (x\x) to (\x\y);

\foreach \x/\r in {1/2,2/3,k/k1}
\foreach \y in {1,2,3,4,5}
\draw[-] (x\r) to (\x\y);

\end{tikzpicture}
\caption{An illustration of the underlying graph from the reduction behind~\Cref{hardness fast strict}.
The edges of the biclique~$(L\cup \{s',s'',t',t''\}, V \cup X)$ are not depicted.
Note that~$L$ is a clique of size~$2k+2$.}
\label{fig fastest graph}
\end{figure}
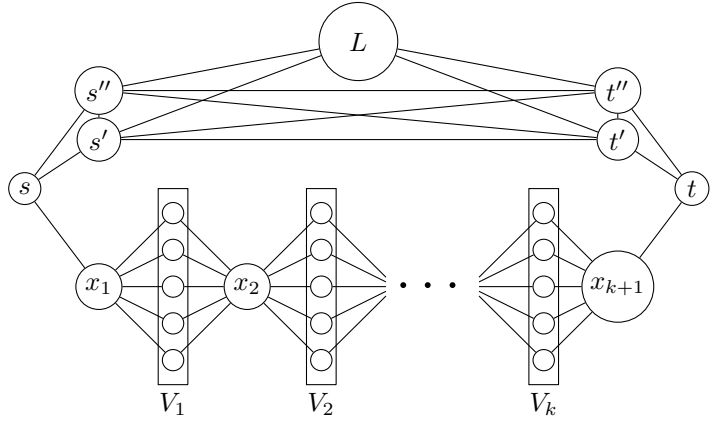

\iflong
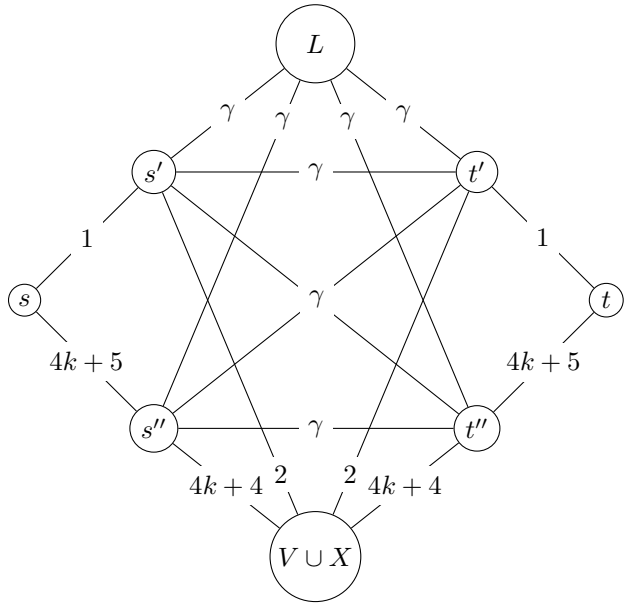
\begin{figure}
\centering
\begin{tikzpicture}[scale=.85]

\tikzstyle{k}=[circle,fill=white,draw=black,minimum size=10pt,inner sep=2pt]

\node[k] (s) at (0,0) {$s$};

\begin{scope}[xshift=9cm]
\node[k] (t) at (0,0) {$t$};
\end{scope}

\node[k] (s1) at (2,2) {$s'$};
\node[k] (s2) at (2,-2) {$s''$};
\node[k] (t1) at (7,2) {$t'$};
\node[k] (t2) at (7,-2) {$t''$};

\node[k,inner sep=7pt] (L) at (4.5,4) {$L$};
\node[k] (VX) at (4.5,-4) {$V \cup X$};

\draw[-] (s) -- (s1) node [midway, fill=white] {$1$};
\draw[-] (s) -- (s2) node [midway, fill=white] {$4k+5$};
\draw[-] (t) -- (t1) node [midway, fill=white] {$1$};
\draw[-] (t) -- (t2) node [midway, fill=white] {$4k+5$};

\draw[-] (s1) -- (t1) node [midway, fill=white] {$\gamma$};
\draw[-] (s1) -- (t2) node [midway, fill=white] {$\gamma$};
\draw[-] (s2) -- (t1) node [midway, fill=white] {$\gamma$};
\draw[-] (s2) -- (t2) node [midway, fill=white] {$\gamma$};

\draw[-] (s1) -- (L) node [midway, fill=white] {$\gamma$};
\draw[-] (s2) -- (L) node [very near end, fill=white] {$\gamma$};
\draw[-] (t1) -- (L) node [midway, fill=white] {$\gamma$};
\draw[-] (t2) -- (L) node [very near end, fill=white] {$\gamma$};

\draw[-] (s1) -- (VX) node [very near end, fill=white] {$2$};
\draw[-] (s2) -- (VX) node [midway, fill=white] {$4k+4$};
\draw[-] (t1) -- (VX) node [very near end, fill=white] {$2$};
\draw[-] (t2) -- (VX) node [midway, fill=white] {$4k+4$};

\end{tikzpicture}
\caption{An illustration of the labeling of the first described time block for a realization of the instance of~\FAST from the reduction behind~\Cref{hardness fast strict}.
For a better readability, $\gamma$ substitutes~$2k+3$ and the addition by the lower label~$\alpha$ of the time block is omitted on all labels.
Only the edges that receive labels are depicted.
}
\label{fig first frame}
\end{figure}
\fi

\subparagraph{Intuition.}
The idea behind the reduction is that realizing all entries besides~$D_{s,t}$ is possible, regardless of whether~$G$ contains a (multicolored) clique of size~$k$.
We will show that this is ensured by the fact that~$G[V_i\cup V_j]$ with~$1\leq i < j \leq k$ is a disjoint union of bicliques.
The difficulty to decide whether the matrix is realizable thus comes from the difficulty of deciding whether the entry~$D_{s,t}$ can additionally be realized, which can only be done by using vertices of~$X\cup V$ as intermediate vertices of the path.
Let~$S$ denote the vertices of~$V$ on any path~$P$ realizing the entry~$D_{s,t}=2k+2$.
Based on the structure of the underlying graph (see~\Cref{fig fastest graph}), $S$ contains for each~$i\in [1,k]$ at least one vertex of~$V_i$.
By definition of the entries between vertices of~$V$, these vertices need to form a clique in the original graph, as only adjacent vertices~$u$ and~$v$ in~$G$ fulfill~$D_{u,v} \leq 2k+2$, which is the duration of~$P$.

\subparagraph{Correctness.}
We now show that~$D$ is realizable if and only if~$G$ admits a clique of size~$k$.
More precisely, we show that if~$G$ admits a clique of size~$k$, then there is a realization for~$D$ with exactly one label per edge.

$(\Rightarrow)$
Let~$\lambda\colon E' \to 2^{\mathbb{N}}$ be an edge labeling, such that~$\mg := (G',\lambda)$ realizes~$D$.
We show that~$G$ has a clique of size~$k$.
Consider the entry~$D_{s,t} = 2k+2$.
Since~$\mg$ is a realization of~$D$, this implies that the fastest temporal path from~$s$ to~$t$ has duration exactly~$2k+2$.
Let~$P$ be an arbitrary fastest temporal path from~$s$ to~$t$ in~$\mg$.
Since the duration of~$P$ is~$2k+2$, the duration of each (not necessarily proper) subpath of~$P$ is at most~$2k+2$.
Hence, for any two distinct vertices~$a,b$ of~$P$, where~$a$ precedes~$b$ in~$P$, the entry~$D_{a,b}$ is at most~$2k+2$, as~$D$ is realized by~$\mg$.
This immediately implies that~$P$ does not visit any vertex of~$\{s',s'',t',t''\} \cup L$, since~$D_{s,t'} = D_{s,t''} = D_{s',t} = D_{s'',t} = 2k+3$ and~$D_{s,\ell} = D_{\ell,t} = 2k+3$ for each vertex~$\ell\in L$.
That is, $P$ only uses vertices of~$V\cup \{s,t\} \cup X$.
By definition, each path from~$s$ to~$t$ in~$G'[V\cup \{s,t\} \cup X]$ traverses all vertices of~$X$ and one vertex of each of the color classes of~$V$, that is, for each~$i\in [1,k]$, the path contains one vertex of~$V_i$.
This in particular holds for~$P$.
Let~$S$ be the vertices of~$V$ that are visited by~$P$.
By the above, $S$ has size at least~$k$.
Moreover, for each two distinct vertices~$a$ and~$b$ of~$S$, $D_{a,b} = D_{b,a} \leq 2k+2$, since the subpath between~$a$ and~$b$ of~$P$ has duration at most~$2k+2$ and~$\mg$ realizes~$D$.
This implies that~$\{a,b\}$ is an edge of~$G$, as otherwise, $D_{a,b} = D_{b,a}$ is defined as~$2k+3$.
Consequently, $S$ is a clique of size~$k$ in~$G$.

$(\Leftarrow)$
Let~$S$ be a clique of size~$k$ in~$G$ and for each~$i\in [1,k]$, let~$v_i$ denote the unique vertex of~$S\cap V_i$.
We define a labeling~$\lambda\colon E' \to \mathbb{N}$, such that~$\mg:= (G',\lambda)$ realizes~$D$.
To this end, we describe several~\emph{time blocks}, that is, intervals~$[\alpha,\beta]$ with~$\beta \geq \alpha + 4k+6$, such that only the described edges receive a label from this interval, and all other edges receive no label from~$[\alpha-(4k+6), \beta + (4k+6)]$.
The reason behind this is that labeling edges in different time blocks do not create paths of duration less than~$4k+6$, which is larger than the largest entry of~$D$.
Hence, we can show that our labeling realizes~$D$ by showing that for each two vertices~$a$ and~$b$ of~$G'$ (i)~there is a time block in which there is a temporal path from~$a$ to~$b$ of duration exactly~$D_{a,b}$ and (ii)~for each time block, there is no temporal path from~$a$ to~$b$ of duration less than~$D_{a,b}$.
Note that the order of time blocks does not matter.
Hence, when describing the labeling, we simply describe a collection of time blocks which in total fulfill the above properties, while not explicitly defining the concrete start and end time of the time blocks.
In the following, we mainly focus on realizing all entries of value at least~$2$ in~$D$. 
Afterwards, we describe how to realize the entries of value~$1$. 
\begin{itemize}
\item \textbf{Realizing all entries involving vertices of~$V' \setminus (V \cup X)$ besides~$D_{s,t}$.}
\iflong
In the first time block~$[\alpha,\beta]$, we realize all entries of~$D$ of value at least~$2$ that involve at least one vertex of~$\{s,s',s'',t,t',t''\} \cup L$, besides the entry~$D_{s,t}$ (see~\Cref{fig first frame}).
We set~$\lambda(\{s,s'\}) := \lambda(\{t,t'\}) := \alpha + 1$, $\lambda(\{s'',s\}) := \lambda(\{t'',t\}) := \alpha + 4k + 5$.
Furthermore, for each vertex~$q\in V\cup X$, we set~$\lambda(\{s',q\}) := \lambda(\{t',q\}) := \alpha + 2$ and~$\lambda(\{q,s''\}) := \lambda(\{q,t''\}) := \alpha + 4k+4$, and we set~$\lambda(\{s',t'\}) := \lambda(\{s',t''\}) := \lambda(\{s'',t'\}) := \lambda(\{s'',t''\}) := \alpha + 2k+3$.
For each vertex~$\ell\in L$ and each vertex~$z\in \{s',t',s'',t''\}$, we also set~$\lambda(\{\ell,z\}) := \alpha + 2k+3$.

This realizes the entries between~$s$ and vertices of~$L$ via the paths~$(s,s',\ell)$ and~$(\ell,s'',s)$ for each~$\ell \in L$ with labels~$(\alpha +1,\alpha +2k+3)$ and~$(\alpha +2k+3,\alpha +4k+5)$.
These paths have duration~$2k+3$ as required by~$D$.
Moreover, no faster paths between~$s$ and a vertex~$\ell\in L$ exists in this time block, since each edge incident with~$\ell$ has label~$\alpha + 2k+3$ and the only labels incident with~$s$ are~$\alpha + 1$ and~$\alpha + 4k+5$.
Analogously, all entries between~$t$ and vertices of~$L$ are realized.
Since~$s$ and~$t$ are the only vertices that are not adjacent to the vertices of~$L$, this thus realizes all entries of value at least~$2$ involving vertices of~$L$.
By the path~$(t,t',s'',s)$ with labels~$(\alpha+1,\alpha+2k+3,\alpha+4k+5)$ we further realize the entry~$D_{t,s}$ and no faster path from~$t$ to~$s$ is realized, since both~$s$ and~$t$ only have the incident labels~$\alpha + 1$ and~$\alpha + 4k+5$ and are not adjacent.
Similar to the paths between~$s$ and vertices of~$L$, there are the paths~$(s,s',\widehat{t})$ and~$(\widehat{t},s'',s)$ for each~$\widehat{t} \in \{t',t''\}$ with labels~$(\alpha +1,\alpha +2k+3)$ and~$(\alpha +2k+3,\alpha +4k+5)$, respectively.
These paths realize the entries between~$s$ and vertices of~$\{t',t''\}$.
Moreover, no faster paths between these vertices are realized by this time block.
This is due to the fact that (i)~each path going over a vertex of~$V\cup X$ in this time block takes time at least~$4k+3$ (since~$\alpha+2$ and~$\alpha+4k+4$ are the only labels incident with vertices of~$V\cup X$), (ii)~each path going over~$t$ has duration at least~$4k+5$, and (iii)~all other paths use an edge of label~$\alpha + 2k+3$ and a label from~$\{\alpha + 1, \alpha + 4k+5\}$. 
Analogously, the entries between~$t$ and the vertices of~$\{s',s''\}$ are realized.
This thus shows that all entries of value at least~$2$ involving any vertex of~$\{s',s'',t',t''\}$ are realized.
For each vertex~$q\in V\cup X$, there are also the paths~$(s,s',q)$, $(t,t',q)$, $(q,s'',s)$, and~$(q,t'',t)$ of duration~$2$.
These paths realize all entries of value larger than~$1$ between~$\{s,t\}$ and vertices of~$V\cup X$.
Thus, all entries of value at least~$2$ involving~$s$ or~$t$ are also realized by this time block, excluding~$D_{s,t}$.
It remains to argue that the time block does not create paths that are too fast between some vertices.
Since we already argued that all entries involving at least one vertex of~$V'\setminus (V\cup X)$ are realized in this time block, we only have to consider entries between vertices of~$V\cup X$ that could theoretically be violated.
Between these vertices, the value of the largest entry is at most~$2k+3$.
No path faster than this can be realized in this time block between vertices of~$V\cup X$, since each such vertex has only the incident labels~$\alpha+2$ and~$\alpha + 4k+4$.
That is, each path in this time block that starts and ends in a vertex of~$V\cup X$ traverses both an edge with label~$\alpha+2$ and an edge with label~$\alpha + 4k+4$, implying that the duration of that path is at least~$4k+3 > 2k+3$.
Thus, this time block realizes all entries of value at least two involving vertices of~$V' \setminus (V \cup X)$ besides~$D_{s,t}$ and does not create any paths that are too fast.
\else
We show that we can realize all these entries by only labeling edges that have at least one endpoint on~$\{s',s'',t',t''\}$.
This proof is deferred to the appendix.
\fi
\item \textbf{Realizing~$D_{s,t}$.}
Next, we define a time block~$[\alpha,\beta]$ that realizes the entry~$D_{s,t}$.
Recall that~$S$ is a clique in~$G$ and that for each~$i\in [1,k]$, $v_i$ denotes the vertex of~$S\cap V_i$.
Consider the path~$P:=(s,x_1,v_1,\dots,x_k,v_k,x_{k+1},t)$ and label the edges of this path with consecutive time labels starting with~$\alpha$.
Hence, this path has duration equals to its length, namely~$D_{s,t}=2k+2$.
\iflong
Since the entries of~$D$ between vertices of~$\{s,t\} \cup X$ and~$V$ are of value at most~$2$ and the entries between vertices of~$\{s,t\} \cup X$ and~$X$ are of value at most~$2$, this creates no path of duration less than~$D_{a,b}$ for vertices~$a$ and~$b$ of~$P$ that are not both from~$V$.
Moreover, the only vertices of~$P$ that are from~$V$ are all from the clique~$S$.
Hence, the entry~$D_{a,b}$ also equals~$2$ for all such vertices~$a$ and~$b$ of~$P$ that are from~$V$.
Thus, this time block does not create paths that are too fast.
\else
The argument that this creates no paths that are too fast is deferred to the appendix. 
\lvnew{It mainly comes from the fact that for any two vertices~$a$ and~$b$ of~$P$ that are both from~$V$, $\{a,b\}\in E$ since $S$ is a clique, and we have $D_{a,b}=2$.}
\fi
\end{itemize}
So far, we realized all entries of~$D$ of value at least~$2$ that involve at least one vertex of~$\{s,s',s'',t,t',t''\} \cup L$.
In the following, we describe further time blocks to realize the entries of~$D$ of value at least~$2$ involving only vertices of~$V\cup X$.
To this end, we will only use edges between~$X\cup V$ and~$L$.
Note that none of these edges has received a label in the previous time blocks, that is, the only edges incident with vertices of~$L$ that received labels so far were the edges between~$L$ and~$\{s',s'',t',t''\}$.
\iflong
Let the vertices of~$L$ be called~$\{\ell^*,\ell^{**}\}\cup \{\ell_i,\ell_i'\mid 1\leq i \leq k\}$.
\else
Let~$L=\{\ell^*,\ell^{**}\}\cup \{\ell_i,\ell_i'\mid 1\leq i \leq k\}$.
\fi
\begin{itemize}
\item \textbf{Realizing entries between vertices of~$V$ of value~$2k+3$ and entries between vertices of~$V$ and~$X$.}
We define a time block~$[\alpha,\beta]$ as follows: 
For each vertex~$v\in V$, we set~$\lambda(\{v,\ell^*\}) := \alpha+1$ and~$\lambda(\{v,\ell^{**}\}) := \alpha+2k+3$.
For each vertex~$x\in X$, we set~$\lambda(\{x,\ell^*\}) := \alpha$ and~$\lambda(\{x,\ell^{**}\}) := \alpha+2k+4$.
Finally, we set~$\lambda(\{\ell^*,\ell^{**}\}) = \alpha+2$.
Note that each vertex of~$V$ has only two incident labels in this time block, namely, $\alpha + 1$ and~$\alpha + 2k+3$.
Hence, no temporal path between vertices of~$V$ has duration less than~$2k+3$.
Moreover, since entries involving a vertex from~$V$ and a vertex from~$X$ are of value at most~$2$, we guarantee that we do not create paths that are too fast in this time block.
We now show that this time block realizes (i)~all entries of value at least~$2$ between vertices of~$V$ and~$X$ and (ii)~all entries between vertices of~$V$ of value~$2k+3$.
For the first type, let~$v\in V$ and~$x\in X$ with~$\{v,x\}\notin E'$.
That is, $D_{x,v}=2$.
Then, there is a temporal path~$(x,\ell^*,v)$ in this time block of duration~2.
Similarly, the temporal path~$(v,\ell^{**},x)$ also has duration~2.
Now consider the second type.
For each two distinct vertices~$u$ and~$v$ of~$V$ with~$D_{u,v}\neq 2$, there is the temporal path~$(u,\ell^*,\ell^{**},v)$ with labels~$(\alpha+1,\alpha+2,\alpha+2k+3)$.
This path has duration~$2k+3 = D_{u,v}$.
This time block realizes the stated entries of~$D$.

\item \textbf{Realizing entries between vertices of~$X$.}
For each~$i\in [1,k]$, we define a time block~$[\alpha_i,\beta_i]$ in which we set~$\lambda(\{x_i,\ell_i\}) := \alpha_i$ and~$\lambda(\{\ell_i,x\}) := \alpha_i+1$ for each~$x\in X \setminus \{x_i\}$.
For each~$i\in [1,k]$, this realizes the entries~$D_{x_i,x}$ with~$x\in X\setminus \{x_i\}$.
Similarly, we add a time block~$[\alpha_{k+1},\beta_{k+1}]$ in which we set~$\lambda(\{x_{k+1},\ell'_1\}) := \alpha_{k+1}$ and~$\lambda(\{\ell'_1,x\}) := \alpha_{k+1}+1$ for each~$x\in X \setminus \{x_{k+1}\}$.
These time blocks realize all entries of~$D$ between vertices of~$X$.

\item \textbf{Realizing entries between vertices of~$V$ of value~$2$.}
Recall that we have to ensure that there is a path of duration 2 between the endpoints of each edge~$e\in E$ in our temporal graph.
To define the necessary time blocks, we will highly rely on the fact that for each~$1\leq a < b \leq k$, $G[V_a\cup V_b]$ is a vertex disjoint union of bicliques.
This property will allow us to realize all edges between~$V_a$ and~$V_b$ via just two vertices of~$L$.
We can do this for several combinations of color classes via the same two vertices of~$L$, as long as no color class occurs in more than one pair.
We formalize this as follows.  
Let~$M_1, \dots, M_k$ be a partition of~$\{(a,b)\mid 1\leq a < b \leq k\}$, such that for each~$i\in [1,k]$ and each~$a\in [1,k]$, there is at most one ordered pair in~$M_i$ that contains~$a$.
That is, $M_i$ is a matching in the directed graph with vertex set~$[1,k]$ and edge set~$\{(a,b)\mid 1\leq a < b \leq k\}$.
Note that such a partition exists due to the fact that a clique on~$k$ vertices has a proper edge coloring with~$k$ colors.
Let~$i\in [1,k]$.
We let~$E_{M_i}$ denote all edges of~$G$ between each pair of color classes in~$M_i$, that is, $E_{M_i} := \bigcup_{(a,b)\in M_i} E(V_a,V_b)$.
Since~$M_i$ is a matching and~$G[V_a\cup V_b]$ is a disjoint union of bicliques for each~$1\leq a < b \leq k$, $G_i := (V,E_{M_i})$ is also a disjoint union of bicliques.
That is, each connected component in~$G_i$ is a biclique.
We use the vertices~$\ell_i$ and~$\ell_i'$ of~$L$ to realize the entries of~$D$ corresponding to the edges of~$E_{M_i}$.
For each connected component of~$G_i$ with bipartition~$(A,B)$, we add a new time block~$[\alpha,\beta]$ and set~$\lambda(\{v_a,\ell_i\}) := \alpha$ for each~$v_a\in A$ and~$\lambda(\{v_b,\ell_i\}) := \alpha+1$ for each~$v_b\in B$.
This realizes paths of duration~$2$ from each vertex of~$A$ to each vertex of~$B$ and no other temporal paths of length more than 1.
Since~$(A,B)$ is a biclique, for all these vertex pairs, the entry in the matrix is also~$2$.
In the same way, we also add a new time block~$[\alpha',\beta']$ and set~$\lambda(\{v_b,\ell_i'\}) := \alpha$ for each~$v_b\in B$ and~$\lambda(\{v_a,\ell_i'\}) := \alpha+1$ for each~$v_a\in A$.
This thus realizes also the entries of duration~$2$ from each vertex of~$B$ to each vertex of~$A$.
Since~$E = \bigcup_{i\in [1,k]} E_{M_i}$, this implies that we realized for each edge~$\{u,v\}\in E$ the entries~$D_{u,v}$ and~$D_{v,u}$ of value~$2$ by the above time blocks.
\end{itemize}

Hence, all entries of value at least 2 in~$D$ are realized by~$\lambda$.
Let~$E''$ denote the edges of~$E'$ that have not received a label yet.
We add one final time block from which all edges of~$E''$ receive the same label.
This surely does not create new temporal paths of length more than 1 for which the duration is at most~$4k+5$.
This completes the definition of~$\lambda$.
Thus, also all entries of value 1 are realized.
By definition of the time blocks, we showed that~$(G',\lambda)$ realizes the input matrix~$D$ even with just a single label per edge.
\end{proof}

Based on this reduction, we can now directly transfer the hardness result to~\PFAST even when allowing at most one label per edge.
\iflong

\prob{\PFAST}
{A duration matrix $D$ of size $n\times n$ and a period~$\Delta$.}
{Is there a~$\Delta$-periodic temporal graph that realizes the fastest path matrix $D$?}

\else
\fi
That is, we simply define the period~$\Delta$ to be an integer much larger than~$n^2 \cdot \max D$, which ensures that all fastest paths start and end within a window of~$\Delta$ consecutive time steps (see~\cite{EMW24}).

\begin{theorem}\label{th:vc}
Even when only allowed to put one label per edge and per period, \PFAST is W[1]-hard when parameterized by the vertex cover number of the underlying graph plus the largest entry of~$D$.
\end{theorem}

This answers an open question by Klobas et al.~\cite{KMMS24} and Erlebach et al.~\cite{EMW24} about the parameterized complexity of the problem with respect to the vertex cover number. 
Furthermore, this reduction improves significantly over the known hardness result for parameter feedback vertex set number.
It also shows that \PFAST can presumably not be solved in FPT time for the combined parameter of the vertex cover number plus~$\ell$ (the number of allowed labels per edge and per period) plus the largest entry in~$D$. 
Thus, in the FPT algorithm by Erlebach et al.~\cite{EMW24} for the vertex cover number plus the period~$\Delta$, one cannot replace~$\Delta$ by~$\ell$ plus the largest entry of~$D$.

\iflong
\paragraph*{Non-Strict fastest paths}
\fi
\iflong Next, we show that a  \else A \fi similar reduction also shows similar intractability results for~\FAST with non-strict paths.
The following reduction however requires more than one label per edge.

\iflong
\begin{theorem}
\else
\begin{theorem}[$\star$]
\fi\label{hardness fast nonstrict}
\NSFAST 
is NP-hard and W[1]-hard when parameterized by the vertex cover number of the underlying graph plus the largest entry of~$D$.
\end{theorem}
\iflong
\begin{proof}
We again reduce from \MCC.

Let~$I:=(G=(V_1\cup\dots\cup V_k, E), k)$ be an instance of~\MCC.
The requirement that the induces subgraph between any two color classes is a disjoint union of bicliques is not required for this reduction.
We obtain an instance~$D$ of \NSFAST as follows:
We take the instance of~\FAST produced by the reduction behind~\Cref{hardness fast strict} and restrict it to the vertices of~$V\cup X \cup \{s,s',t,t'\}$.
Let~$G'$ be the respective underlying graph.
Since~$V$ is an independent set in the instance from the reduction behind~\Cref{hardness fast strict}, $X \cup \{s',t'\}$ is a vertex cover of size~$k+3$ of~$G'$.

This completes the construction.
Intuitively, the main difficulty comes from realizing the entry~$D_{s,t}$ while not violating the other entries.
We show that this can only be realized if and only if~$G$ has a clique of size~$k$.

$(\Rightarrow)$
The proof of this direction is identical to the one on the proof of~\Cref{hardness fast strict}.
We thus only recall the main parts of the proof for this direction.
\begin{itemize}
\item To realize the entry~$D_{s,t}$ of value~$2k+2$, there needs to be a temporal path~$P$ of duration exactly~$2k+2$.
\item For each pair $(u,v)$ of vertices on~$P$, where~$u$ appears before~$v$ on~$P$, it holds that~$D_{u,v} \leq 2k+2$.
\item By~$D_{s,t'} = D_{s',t} = 2k+3$, $P$ contains neither~$s'$ nor~$t'$.
\item Each path form~$s$ to~$t$ in~$G' - \{s',t'\} := G'[V \cup X \cup \{s,t\}]$ contains exactly one vertex of~$V_i$ for each~$i\in [1,k]$. Thus, $P$ visits for each~$i\in [1,k]$, exactly one vertex~$v_i$ of~$V_i$.
\item Since~$D_{v_i,v_j} \leq 2k+2$ for each~$i\in [1,k-1]$ and each~$j\in [i+1,k]$, there is the edge~$\{v_i,v_j\}$ in~$G$ (as otherwise, $D_{v_i,v_j} = D_{v_j,v_i} = 2k+3$).
\item Thus, $\{v_i\mid i\in [1,k]\}$ is a clique of size~$k$ in~$G$.
\end{itemize}

$(\Leftarrow)$
Let~$S$ be a clique of size~$k$ in~$G$ and for each~$i\in [1,k]$, let~$v_i$ denote the unique vertex of~$S\cap V_i$.
We define a labeling~$\lambda\colon E' \to 2^\mathbb{N}$, such that~$\mg:= (G',\lambda)$ realizes~$D$.
We present the labeling again via the help of time blocks.

Let~$(p,q)\in (V \cup X)^2$ be a pair of distinct vertices, such that~$\{p,q\}$ is not an edge of~$G'$.
We add a new time block~$[\alpha,\beta]$ in which we add labels~$\alpha + 1$ and~$\alpha + 2 \cdot D_{p,q} - 1$ to edge~$\{p,s'\}$ and add label~$\alpha + D_{p,q}$ to edge~$\{s',q\}$.
This realizes both entries~$D_{p,q}$ and~$D_{q,p}$ and no other paths of duration more than 1 exists in this time block.

For each vertex~$p\in V \cup X$, we add a time block~$[\alpha,\beta]$, in which we add labels~$\alpha + 1$ and~$\alpha + 3$ to the edge~$\{p,s'\}$ and add label~$\alpha + 2$ to edge~$\{s,s'\}$.
This creates paths of duration exactly~$2$ between~$p$ and~$s$ and no other paths of duration more than 1 exists in this time block.
Similarly, we ensure paths of duration between~$p$ and~$t$ by adding a time block~$[\alpha',\beta']$, in which we add labels~$\alpha' + 1$ and~$\alpha' + 3$ to the edge~$\{p,t'\}$ and add label~$\alpha' + 2$ to edge~$\{t,t'\}$.
All these time blocks together realize all entries between~$V$ and~$\{s,t\}$ of value more than~$1$.

It remains to realize the entries between the vertices of~$\{s,s',t,t'\}$ of value at least~$2$.
To this end, we add two time blocks.
In the first time block~$[\alpha,\beta]$, we add label~$\alpha + 1$ and~$\alpha + 4k+5$ to both edges~$\{s,s'\}$ and~$\{t,t'\}$, and add label~$\alpha + 2k+3$ to edge~$\{s',t'\}$.
This ensures paths of duration~$2k+3$ between~$s$ and~$t'$ and between~$t$ and~$s'$.
Moreover, this ensures that there is a path of duration~$D_{t,s} = 4k+5$ from~$t$ to~$s$.
Observe that no faster paths between these vertices are realized.
Thus, the only entry of value at least~$2$ that is not yet realized is the entry~$D_{s,t}$ of value~$2k+2$.
We do this in the same way as in the proof of~\Cref{hardness fast strict}.
Due to the similarity, we just recall the rough idea.
Recall that~$S$ is a clique in~$G$ and that for each~$i\in [1,k]$, $v_i$ denotes the vertex of~$S\cap V_i$.
Consider the path~$P:=(s,x_1,v_1,\dots,x_k,v_k,x_{k+1},t)$ and label the edges of this path with consecutive time labels starting with~$\alpha$.
Hence, this path has duration equals to its length, namely~$D_{s,t}=2k+2$.
Since~$S$ is a clique in~$G$, all entries of~$D$ between these vertices are at most~$2$. 
Thus, this time block does not create paths that are too fast.

Finally, for each entry~$D_{p,q}$ of value~$1$ for which the edge~$\{p,q\}$ has not received a label so far, we add a time block from which the edge~$\{p,q\}$ receives an arbitrary label.
This thus also realizes all entries of value~$1$, which in total implies that~$D$ is realized by~$(G',\lambda)$.
\end{proof}

By replacing the vertices~$s'$ and~$t'$ by $\Theta(n^2)$~true twins each (also with respect to all respecting entries in the matrix), we can obtain an equivalent instance of~\NSFAST, such that each constructed yes-instance in that way can be realized with only one label per edge.
This, however, has a vertex cover of unbounded size and the parameterized hardness does not transfer.

\begin{corollary}\label{hardness fast nonstrict simple}
\NSFAST 
is NP-hard and W[1]-hard when parameterized by the largest entry of~$D$.
This holds even on a family of instances for which all yes-instances are realizable with only one label per edge.
\end{corollary}

This shows that \PNSFAST 
remains NP-hard. 
Moreover, since the reduction behind~\Cref{hardness fast nonstrict} never had to make use of non-strict paths to realize the obtained yes-instances, but only used paths that were also strict, this hardness result would also transfer to the version of~\FAST, where we require the resulting temporal graph to be~\emph{proper}, that is, no two adjacent edges are allowed to share a label.
\fi

\section{Shortest paths}\label{sec short}
In this section, we consider the question for~\emph{shortest} temporal paths.

\prob{\SHORT}
{A distance matrix $D$ of size $n\times n$.}
{Is there a temporal graph $\mg$ such that~$\Sh(\G)=D$?}

Note that a realization of~$D$ can only assign labels to edges~$\{u,v\}$ where~$D_{u,v} = D_{v,u} = 1$.
Hence, $G=([n],E)$ with~$E:=\{\{u,v\}\mid D_{u,v} = 1 \land D_{v,u} = 1\}$ is the underlying graph of every realization of~$D$.
\iflong

\fi
We show the NP-hardness of both the strict and the non-strict variants.

\iflong
\begin{theorem}
\else
\begin{theorem}[$\star$]
\fi\label{hardness shortest}
\SHORT and \NSSHORT are NP-hard.
\end{theorem} 

\iflong
\begin{proof}
\else
\begin{proof}[Proof (sketch).]
\fi
\iflong
We reduce from the NP-hard~\SAT~\cite{K72}.
\else
We reduce from \SAT.
\fi
Let~$F$ be an instance of~\SAT where each variable occurs at least once positively and at least once negatively, and where no clause contains the same variable both positively and negatively.
\iflong
Under these restrictions, \SAT is NP-hard~\cite{K72}.
\fi

\textbf{Construction.}
Let~$X$ be the variable set of~$F$ and let~$C$ denote the clauses of~$F$.
To obtain an instance~$D$ of~\SHORT or~\NSSHORT, we first define the underlying graph~$G=(V,E)$ that contains an edge~$\{u,v\}$ if and only if~$D_{u,v}=D_{v,u}=1$ (see~\Cref{fig shortest}).
The graph~$G$ contains for each variable~$x\in X$  the vertices~$x$ and~$\overline{x}$ which are joined by an edge.
For each clause~$c\in C$, we also add a vertex~$c$, which we make adjacent to all vertices corresponding to literals that are contained in~$c$.
Additionally, we add three more vertices to~$G$.
A vertex~$v^*$ which is adjacent to all vertices of~$C$, and two vertices~$\top$ and~$\bot$ that are adjacent to all vertices representing literals, that is, to the vertices of~$\{x,\overline{x}\mid x\in X\}$.

Next, we describe the remaining entries of~$D$.
Let~$c$ be a clause of~$C$.
We set~$D_{c,\top} := 3$ and~$D_{\top,c} := D_{c,\bot} := D_{\bot,c} := 2$.
For each other clause~$c'$ of~$C$, we set~$D_{c,c'} := D_{c',c} := 2$.
For each positive literal~$x$ that occurs in~$c$, we set~$D_{c,\overline{x}} := D_{\overline{x},c} := 2$.
Similarly, for each negative literal~$\overline{x}$ that occurs in~$c$, we set~$D_{c,x} := D_{x,c} := 2$.
For each variable~$x$ for which neither~$x$ nor~$\overline{x}$ occurs in~$c$, we set~$D_{c,\overline{x}} := D_{\overline{x},c} := D_{c,x} := D_{x,c} := 3$.
This defines all entries regarding vertices of~$C$.

Let~$\ell_1$ and~$\ell_2$ be distinct literals, such that they are not the negation of each other.
We set~$D_{\ell_1,\ell_2} := D_{\ell_2,\ell_1} := 2$.
For each literal~$\ell$, we also set~$D_{\ell,v^*} := D_{v^*,\ell} := 2$.

Finally, we set~$D_{v^*,\bot} := D_{\bot,v^*} := 3$, $D_{\bot,\top} := D_{\top,\bot} := 2$, $D_{v^*,\top} := 4$, and~$D_{\top,v^*} := 3$.
This completes the definition of~$D$.

Note that nearly all defined entries are the exact distances between the vertices in the underlying graph~$G$.
The only exceptions are the entry~$D_{v^*,\top}$ and the entry~$D_{c,\top}$ for each clause~$c\in C$.
Hence, only for the vertex pairs~$(\{v^*\} \cup C) \times \{\top\}$, one could possibly create a temporal path that has length less than the respective entry of~$D$.
Based on this property, we can prove that a labeling~$\lambda \colon E \to 2^{\mathbb{N}}$ realizes~$D$ by showing the following two points:
\begin{itemize}
\item For distinct vertices~$a$ and~$b$ of~$V$, there is a temporal path of length~$D_{a,b}$ from~$a$ to~$b$.
\item For each~$a\in \{v^*\} \cup C$, there is no temporal path of length less than~$D_{a,\top}$ from~$a$ to~$\top$.
\end{itemize}

\begin{figure}
\centering
\scalebox{.75}{
\begin{tikzpicture}[scale=1]

\tikzstyle{k}=[circle,fill=white,draw=black,minimum size=10pt,inner sep=2pt]

\node[k] (t) at (0,3) {$\top$};
\node[k] (b) at (2,3) {$\bot$};

\iflong
\node[k] (v) at (0,-3) {$v^*$};
\else
\node[k] (v) at (0,-2) {$v^*$};
\fi

\node[k] (x)  at (-4.5,1) {$x$};
\node[k] (nx) at (-2.5,1) {$\overline{x}$};

\node[k] (y)  at (-1,1) {$y$};
\node[k] (ny) at (1,1) {$\overline{y}$};

\node[k] (z)  at (2.5,1) {$z$};
\node[k] (nz) at (4.5,1) {$\overline{z}$};

\node[k] (c1)  at (-3.5,-1) {$c_1$};
\node[k] (c2)  at (0,-1) {$c_2$};
\node[k] (c3)  at (3.5,-1) {$c_3$};

\draw[thick, -, dotted] (v) -- (c1);
\draw[thick, -, dotted] (v) -- (c2);
\draw[thick, -, dotted] (v) -- (c3);

\draw[thick, -, dotted] (x) -- (b);
\draw[thick, -, dotted] (nx) -- (b);
\draw[thick, -, dotted] (y) -- (b);
\draw[thick, -, dotted] (ny) -- (b);
\draw[thick, -, dotted] (z) -- (b);
\draw[thick, -, dotted] (nz) -- (b);

\draw[-] (x) -- (t) node [near start, fill=white] {$5$};
\draw[-] (nx) -- (t) node [near start, fill=white] {$2$};
\draw[-] (y) -- (t) node [near start, fill=white] {$2$};
\draw[-] (ny) -- (t) node [near start, fill=white] {$5$};
\draw[-] (z) -- (t) node [near start, fill=white] {$5$};
\draw[-] (nz) -- (t) node [near start, fill=white] {$2$};

\draw[-] (x) -- (nx) node [midway, fill=white] {$1,4,7$};
\draw[-] (y) -- (ny) node [midway, fill=white] {$1,4,7$};
\draw[-] (z) -- (nz) node [midway, fill=white] {$1,4,7$};

\draw[-] (c1) -- (x)  node [near start, fill=white] {$6$};
\draw[-] (c2) -- (nx) node [near start, fill=white] {$3$};
\draw[-] (c1) -- (y)  node [near start, fill=white] {$3$};
\draw[-] (c2) -- (y)  node [near start, fill=white] {$3$};
\draw[-] (c1) -- (z)  node [near start, fill=white] {$6$};
\draw[-] (c3) -- (ny)  node [near start, fill=white] {$6$};
\draw[-] (c3) -- (nz)  node [near start, fill=white] {$3$};
\end{tikzpicture}
}
\caption{An example of the reduction behind~\Cref{hardness shortest} for the formula~$(x \lor y \lor z) \land (\overline{x} \lor y) \land (\overline{y} \lor \overline{z})$.
A labeling that realizes the matrix~$D$ is depicted, where the dashed arcs receive the label set~$\{1,2,7,8\}$.
Moreover this labeling corresponds to a satisfying truth assignment ($x$ = False,  $y$ = True, $z$ = False).}
\label{fig shortest}
\end{figure}
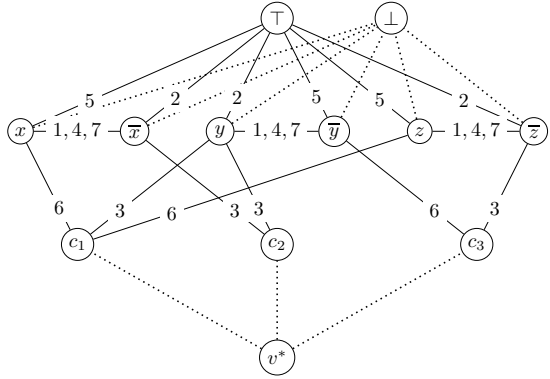

\textbf{Intuition.}
\iflong
Similar to the reduction behind~\Cref{range size 2 hard}, we
\else
We
\fi
have clause vertices that aim to reach~$\top$. 
For each such clause vertex~$c_i$, we want that the shortest temporal path to~$\top$ has length exactly~$3$.
By the structure of the implicit underlying graph, these paths must be of the form~$(c_i,\ell,\overline{\ell},\top)$ for some literal~$\ell$ that occurs in clause~$c_i$.
If two clauses try to use the same variable gadget from different sides to realize their entries, that is, if~$(c_i,\ell,\overline{\ell},\top)$ and~$(c_j,\overline{\ell},\ell,\top)$ are both temporal paths in our solution graph, then in fact at least one of~$(c_i,\ell,\top)$ or~$(c_j,\overline{\ell},\top)$ is also a temporal path, implying that for at least one of the clauses, the shortest temporal path has length~$2$, which is lower than the desired length of~$3$.
Intuitively, this means that in each solution, each variable gadget can only be used in one direction for paths between clauses and~$\top$, which then encodes a satisfying truth assignment.
\iflong

\textbf{Correctness.}
We now show that~$F$ is satisfiable if and only if~$D$ is realizable.

$(\Leftarrow)$
Let~$\mg := (G,\lambda)$ be a realization of~$D$.\footnote{Recall that a realization of~$D$ can only assign label to the edges of~$E$, as all other vertex pairs have a pairwise distance strictly greater than one according to~$D$.} 
We show that~$F$ is satisfiable.
To this end, we define the following truth assignment~$\pi$:
We assign the literal~$x$ to True if and only if~$\max \lambda(\{x,\top\}) < \max \lambda(\{\overline{x},\top\})$.

We now show that the so defined truth assignment satisfies~$F$.
To this end, let~$c$ be an arbitrary clause of~$C$.
We show that one literal contained in~$c$ is assigned to True by~$\pi$.
Since~$\mg$ realizes~$D$ and~$D_{c,\top} = 3$, the shortest temporal path in~$\mg$ from~$c$ to~$\top$ has length exactly~$3$.
Let~$P=(c,a,b,\top)$ be such a shortest temporal path from~$c$ to~$\top$ in~$\mg$ for vertices~$a$ and~$b$.
The vertex~$a$ is a literal~$\ell$ that is contained in the clause~$c$.
This is due to the fact that the only other neighbor of~$c$ in~$G$ is vertex~$v^*$ for which the shortest temporal path towards~$\top$ is supposed to have length at~$4 = D_{v^*,\top}$.
Let~$\overline{\ell}$ denote the complementary literal of~$\ell$.
Now consider the vertex~$b$.
Since~$b$ is the direct predecessor of~$\top$ in~$P$, $b$ is a neighbor of both~$\ell$ and~$\top$ in~$G$.
By definition, $\overline{\ell}$ is the only such vertex.
Thus, $P = (c,\ell,\overline{\ell},\top)$.
We show that~$\max \lambda(\{\ell,\top\}) < \max \lambda(\{\overline{\ell},\top\})$, which then implies that literal~$\ell$ is set to True under~$\pi$.
To this end, let~$\alpha := \min \lambda(\{c,\ell\})$.
We distinguish between the strict and the non-strict case:

In the strict case, $\max\lambda(\{\overline{\ell},\top\}) > \alpha$, since~$P$ is a temporal path in~$\mg$.
Moreover, $\max \lambda(\{\ell,\top\}) \leq \alpha$, as otherwise, $(c,\ell,\top)$ would be a temporal path of length~$2<D_{c,\top}$ in~$\mg$.
Thus, $\max \lambda(\{\ell,\top\}) \leq \alpha < \max\lambda(\{\overline{\ell},\top\})$.

In the non-strict case, $\max\lambda(\{\overline{\ell},\top\}) \geq \alpha$, since~$P$ is a temporal paths in~$\mg$.
Moreover, $\max \lambda(\{\ell,\top\}) < \alpha$, as otherwise, $(c,\ell,\top)$ would be a temporal path of length~$2<D_{c,\top}$ in~$\mg$.
Thus, $\max \lambda(\{\ell,\top\}) < \alpha \leq \max\lambda(\{\overline{\ell},\top\})$.

In both cases, $\max \lambda(\{\ell,\top\}) < \max\lambda(\{\overline{\ell},\top\})$.
This implies that the literal~$\ell$ is set to True under~$\pi$.
Thus, $\pi$ satisfies~$c$.
Since the clause~$c$ was chosen arbitrarily, this implies that~$\pi$ satisfies all clauses of~$C$ and thus the whole formula~$F$.

$(\Rightarrow)$
Assume that~$F$ is satisfied and let~$L$ denote the set of literals that are assigned to True in a fixed but arbitrary satisfying truth assignment for~$F$.
We define a labeling~$\lambda\colon E \to 2^{\mathbb{N}}$ as follows:
For each edge~$e$ incident with one of~$v^*$ or~$\bot$, we set~$\lambda(e) := \{1,2,7,8\}$.
Let~$x$ be a variable of~$X$.
We set~$\lambda(\{x,\overline{x}\}) := \{1,4,7\}$.
If~$x\in L$, that is, if~$x$ is assigned to True, we set~$\lambda(\{\overline{x},\top\}) := \{5\}$, $\lambda(\{x,\top\}) := \{2\}$, $\lambda(\{\overline{x},c\}) := \{6\}$ for each clause~$c\in C$ that contains the literal~$\overline{x}$, and $\lambda(\{x,c\}) := \{3\}$ for each clause~$c\in C$ that contains the literal~$x$.
Otherwise, that is, if~$\overline{x}\in L$, we set~$\lambda(\{\overline{x},\top\}) := \{2\}$, $\lambda(\{x,\top\}) := \{5\}$, $\lambda(\{\overline{x},c\}) := \{3\}$ for each clause~$c\in C$ that contains the literal~$\overline{x}$, and $\lambda(\{x,c\}) := \{6\}$ for each clause~$c\in C$ that contains the literal~$x$.
This completes the definition of~$\lambda$.

We now show that~$\mg := (G,\lambda)$ realizes~$D$.
As discussed previously, we can do this in two steps. 
\begin{itemize}
\item We show that for each two distinct vertices~$a$ and~$b$ of~$V$, there is a temporal path of length exactly~$D_{a,b}$ from~$a$ to~$b$, and
\item we show that for each~$a\in \{v^*\} \cup C$, there is no temporal path of length less than~$D_{a,\top}$ from~$a$ to~$\top$.
\end{itemize}
The latter property follows from the definition of~$\lambda$:
For each literal~$\ell$, $\max \lambda(\{\ell,\top\}) < \min \lambda(\{\ell,c\})$ for each clause~$c\in C$ that contains the literal~$\ell$.
Thus, for each such clause~$c$, $(c,\ell,\top)$ is not a temporal path, which implies that the shortest temporal path from~$c$ to~$\top$ has length at least~$3 = D_{c,\top}$, since each common neighbor of~$c$ and~$\top$ is a literal~$\ell$ that is contained in~$c$.
Note that this holds in both the strict and the non-strict setting.
Similarly, since the neighborhood of~$v^*$ in~$G$ are exactly the clauses of~$C$ and since no such vertex~$c\in C$ has a temporal path of length less than~$3$ towards~$\top$, there is no temporal path of length less than~$4 = D_{v^*,\top}$ from~$v^*$ to~$\top$.
Consequently, for any two vertices~$a$ and~$b$ of~$V$, there is no temporal path from~$a$ to~$b$ of length less than~$D_{a,b}$ in~$\mg$.

It thus remains to show that there always is at least one temporal path from~$a$ to~$b$ in~$\mg$ of length exactly~$D_{a,b}$.
First, we consider all entries involving vertices of~$C$.
\begin{itemize}
\item If~$a$ and~$b$ are from~$C$, the entry~$D_{a,b} = 2$ is realized via the path~$(a,v^*,b)$ which has labels~$(1,2)$.
\item Let~$a$ be from~$C$ and let~$b$ be a literal~$\ell$.
If~$\ell$ is contained in clause~$a$, the two vertices are adjacent and the entries~$D_{a,b}$ and~$D_{b,a}$ of value~$1$ are realized by the direct edge.
If the negation~$\overline{\ell}$ of~$\ell$ is contained in clause~$a$, then  the path~$(a,\overline{\ell},\ell)$ with labels~$(\alpha,7)$ and the path~$(\ell,\overline{\ell},a)$ with labels~$(1,\alpha)$ realize the entries~$D_{a,\ell}$ and~$D_{\ell,a}$ of value~$2$, where~$\alpha$ is the unique label of edge~$\{a, \overline{\ell}\}$ which is an element of~$\{3,6\}$.
If neither~$\ell$ nor the negation of~$\ell$ occurs in clause~$a$, let~$\ell'$ be an arbitrary literal which is contained in clause~$a$.
Then, the path~$(a,\ell',\bot,\ell)$ with labels~$(\alpha, 7,8)$ and the path~$(\ell,\bot,\ell',a)$ with labels~$(1,2, \alpha)$ realize the entries~$D_{a,\ell}$ and~$D_{\ell,a}$ of value~$3$, where~$\alpha$ is the unique label of edge~$\{a, \overline{\ell}\}$ which is an element of~$\{3,6\}$.
Note that the latter case also implies the existence of temporal paths of length~$2$ between~$a$ and~$\bot$.
\item Let~$a$ be from~$C$ and let~$b$ be~$\top$.
Then, since the clause~$c$ is satisfied by the truth assignment, there is some literal~$\ell$ that is contained in~$c$ and also assigned to True.
Let~$\overline{\ell}$ be the negation of~$\ell$.
By definition of~$\lambda$, the paths~$(a,\ell,\overline{\ell},\top)$ has labels~$(3,4,5)$ and the path~$(\top,\ell,a)$ has labels~$(2,3)$.
These two paths realize the entry~$D_{a,\top}$ of value~$3$ and~$D_{\top,a}$ of value~$2$. 
\end{itemize}
Thus, all entries involving at least one vertex of~$C$ are realized.
We continue with the entries involving vertices that correspond to literals.
\begin{itemize}
\item Let~$a$ and~$b$ be literals that are not adjacent in~$G$.
Then, the entries~$D_{a,b}$ and~$D_{b,a}$ of value~$2$ are realized by the paths~$(a,\bot,b)$ and~$(b,\bot,a)$ with labels~$(1,2)$ each.
\item Let~$a$ be a literal and let~$b$  be~$v^*$.
Then, there is some clause~$c$ such that~$c$ contains literal~$a$.
Thus, the entries~$D_{a,b}$ and~$D_{b,a}$ of value~$2$ are realized by the paths~$(a,c,v^*)$ and~$(v^*,c,a)$ with labels~$(\alpha,7)$ and~$(1,\alpha)$ respectively, where~$\alpha$ denotes the unique label of edge~$\{c,a\}$ which is an element of~$\{3,6\}$.
\end{itemize}
Note that these are all remaining entries of value at least~$2$ that involve vertices that correspond to literals.
It remains to consider the entries between the vertices~$v^*$, $\bot$ and~$\top$.
The entries~$D_{\bot,\top}$ and~$D_{\top,\bot}$ of value~$2$ are realized via the paths~$(\bot,\ell,\top)$ and~$(\top,\ell,\bot)$ with labels~$(1,2)$ and~$(2,7)$ respectively, where~$\ell$ is an arbitrary literal that is assigned to True.
For the remaining four entries involving~$v^*$ and a vertex~$b\in \{\top, \bot\}$, consider the previously discussed paths between~$b$ and an arbitrary clause vertex~$c\in C$.
These paths only started/ended at~$c$ with a label between~$3$ and~$6$. 
Thus, we can extend them by a preceding (resp. succeeding) edge from (resp. to) $v^*$ with label~$1$ (resp. $8$) to realize the remaining four entries involving vertex~$v^*$.

Thus, for each two vertices~$a$ and~$b$ of~$V$, the shortest temporal path in~$\mg$ from~$a$ to~$b$ has length exactly~$D_{a,b}$.
Consequently, $\mg$ realizes~$D$.
\else
The correctness is deferred to the appendix.
\fi
\end{proof}

Note that a realization question for shortest temporal paths in periodic temporal graphs is polynomial time solvable. 
Due to the periodicity, a shortest path in the underlying static graph will always be a shortest temporal path in the periodic temporal graph where each edge receives at least one label.
Since the latter is mandatory, the resulting problem is answered with yes if and only if the given matrix~$D$ is the distance matrix of the underlying graph.

\section{Conclusion}\label{sec conc}
Our work spawns several interesting future questions.
\begin{itemize}
\item We showed that \EARLY is polynomial-time solvable if we are allowed to assign up to~$n$ labels per edge but becomes NP-hard when allowing only a single label per edge. What is the smallest number of labels per edge for which this problem is still polynomial-time solvable?
For example, is there an efficient algorithm when we are allowed to assign only~$\frac{n}{2}$ labels per edge?
\item Our hardness results for \SHORT use multiple labels per edge. 
Does this problem become polynomial-time solvable, if we restrict the respective labeling? 
For example, what if we enforce that at most one label per edge is allowed or we require a proper labeling, that is, a labeling where no two adjacent edges share a label?
\item Are there structural parameters for which we can solve~\SHORT in FPT-time.
For example, can we solve the problem efficiently if the underlying graph has bounded treewidth?
\item One could consider approximation of the considered problems. For example under the measurement of fulfilling as many entries as possible.
Are there constant factor approximations for \SHORT or \FAST?
\end{itemize}

\iflong
\else
\clearpage
\fi
\ifhal
\bibliographystyle{plainurl}
\else
\bibliographystyle{plain}
\fi
\bibliography{bib}

\end{document}